\documentclass[11pt]{article}

\usepackage{fullpage}
\usepackage[utf8]{inputenc}
\usepackage{todonotes}
\usepackage[linesnumbered,ruled,vlined]{algorithm2e}
\SetKw{Break}{break}
\SetKw{AND}{and}
\SetKw{OR}{or}

\usepackage{graphics}
\usepackage[dvips]{epsfig}
\usepackage{comment}
\usepackage{xcolor}
\usepackage{amsmath}
\usepackage{amssymb}
\usepackage{amsfonts}
\usepackage{graphicx, xspace}

\newcommand{\ie}{{i.e.,}\xspace}

\newcommand{\true}{\mbox{\it true}}
\newcommand{\false}{\mbox{\it false}}

\newtheorem{theorem}{Theorem}[section]
\newtheorem{lemma}{Lemma}[section]

\newtheorem{claim}{Claim}[section]
\newtheorem{proposition}{Proposition}[section]

\newcommand{\qed}{\hfill $\Box$ \bigbreak}
\newenvironment{proof}{\noindent {\bf Proof.}}{\qed}

\newenvironment{proofclaim}{\noindent{\bf Proof of the claim.}}{\hfill$\star$}

\title{{\bf Almost-Optimal Deterministic Treasure Hunt\\ in Arbitrary Graphs}}
\author{
S\'{e}bastien Bouchard\thanks{Univ. Bordeaux, Bordeaux INP, CNRS, LaBRI, UMR5800, F-33400 Talence, France. E-mail: sebastien.bouchard@u-bordeaux.fr}
\and
Yoann Dieudonn\'{e}\thanks{
MIS Lab., Universit\'{e} de Picardie Jules Verne, France. E-mail: yoann.dieudonne@u-picardie.fr}
\and
Arnaud Labourel\thanks{
Aix Marseille Univ, Université de Toulon, CNRS, LIS, Marseille, France. Email: arnaud.labourel@lis-lab.fr}
\and
Andrzej Pelc\thanks{D\'{e}partement d'informatique, Universit\'{e} du Qu\'{e}bec en Outaouais,
Gatineau, Qu\'{e}bec J8X 3X7,
Canada. E-mail: pelc@uqo.ca.
Supported in part by NSERC discovery grant 2018-03899
and by the Research Chair in Distributed Computing of
the Universit\'{e} du Qu\'{e}bec en Outaouais.}
}
\date{ }

\begin{document}

\baselineskip  0.20in 
	\parskip     0.1in 
	\parindent   0.0in 
\maketitle
\begin{abstract}
A mobile agent navigating along edges of a simple connected graph, either finite or countably infinite, has to find an inert target (treasure) hidden in one of the nodes. This task is known as treasure hunt. The agent  has no {\em a priori} knowledge of the graph, of the location of the treasure or of the initial distance to it.  The cost of a treasure hunt algorithm is the worst-case number of edge traversals performed by the agent until finding the treasure. 
Awerbuch, Betke, Rivest and Singh \cite{ABRS} considered graph exploration and treasure hunt for finite graphs in a restricted model where the agent has a fuel tank that can be replenished only at the starting node $s$. The size of the tank is $B=2(1+\alpha)r$, for some positive real constant $\alpha$, where $r$, called the radius of the graph, is the maximum distance from $s$ to any other node. The tank of size $B$ allows the agent to make at most {$\lfloor B\rfloor$} edge traversals between two consecutive visits at node $s$. 

Let $e(d)$ be the number of edges whose at least one extremity is at distance less than $d$ from $s$. 
Awerbuch, Betke, Rivest and Singh \cite{ABRS}  conjectured that it is impossible to find a treasure hidden in a node at distance at most $d$ at cost nearly linear in $e(d)$. We first design a deterministic treasure hunt algorithm working in the model without any restrictions on the moves of the agent at cost $\mathcal{O}(e(d) \log d)$, and then show how to modify this algorithm to work in the model from \cite{ABRS} with the same complexity. Thus we
refute the above twenty-year-old conjecture. We observe that no treasure hunt algorithm can beat cost $\Theta(e(d))$ for all graphs  and thus 
 our algorithms are also almost optimal.

\vspace*{1cm}

{\bf keywords:} treasure hunt, graph, mobile agent

\vspace*{2cm}

\end{abstract}

\pagebreak

\section{Introduction}

\subsection{The background}

A mobile agent has to find an inert target (treasure) in some environment that can be a network modeled by a graph or a terrain in the plane. This task, known as {\em treasure hunt}, has important applications when the environment is dangerous for humans. When a miner is lost in a contaminated mine, it may have to be found by a robot, and the length of the robot's trajectory should be as short as possible, in order to minimize rescuing time.
In this example, a graph models the corridors of the mine with nodes representing crossings. Another application of treasure hunt in graphs is searching for a data item in a communication network modeled by a graph.
%

\subsection{The models and the problem}

{We consider a simple connected undirected locally finite graph $\mathcal{G}=(V_\mathcal{G},E_\mathcal{G})$, i.e., a graph with nodes of finite degrees. Such a graph can be either finite or countably infinite.} A mobile agent (robot) starts at a node $s$ of $\mathcal{G}$, called the {\em source node}, and moves along its edges. The maximum distance of any node from $s$ is denoted by $r$ and called the {\em radius} of the graph (the radius of countably infinite graphs is infinite).
We make the same assumption as in \cite{ABRS} that the agent has unbounded memory and can recognize already visited nodes and traversed edges. This is formalized as follows. Nodes of $\mathcal{G}$ have distinct labels that are positive integers. Each edge has ports at both of its extremities. Ports corresponding to edges incident to a node of degree $\delta$ are numbered $0,1,\dots, \delta -1$ in an arbitrary way. At the beginning, the agent situated at node $s$ sees its degree. The agent executes a deterministic algorithm: at each step, it selects a port number on the basis of currently available information, and traverses the corresponding edge. When the agent enters the adjacent node, it learns its label, its degree, and the incoming port number.

Each node of $V_\mathcal{G}$ will be identified with its label, and each edge of $E_\mathcal{G}$ will be identified as the quadruple $(v,w,p,q)$, where $v<w$ are labels of the edge extremities, $p$ is its port number at node $v$ and $q$ is its port number at node $w$.

The above simple model will be called {\em unrestricted}. However, some authors imposed additional restrictions, in the case when the graph is finite. The authors of \cite{ABRS} used a restriction of moves of the agent that we will call the {\em fuel-restricted model}. They  assumed that the agent has a fuel tank that can be replenished only at the starting node $s$ of the agent. The size of the tank is $B=2(1+\alpha)r$, for some positive real constant $\alpha$, where $r$ is the radius of the graph. The tank of size $B$ allows the agent to make at most {$\lfloor B\rfloor$} edge traversals between two consecutive visits at node $s$. The restriction used in \cite{DKK} was of a different kind. We will call it the {\em rope-restricted model}. {It was assumed in \cite{DKK} that the agent is {\em tethered}, i.e.,  attached to $s$ by a rope that it unwinds by a length 1 with every forward edge traversal and rewinds by a length of 1 with every backward edge traversal. The rope is infinitely extendible but has to satisfy the following constraint: the segment of the rope unwinded by the agent must never be longer than $L=(1+\alpha)r$, for some positive real constant $\alpha$.} Hence the agent is forced to match every forward edge traversal of an edge with a backward edge traversal, rewinding the rope, in a first-in last-out stack order.

The task of treasure hunt, in any of the above three models, is defined as follows. An adversary hides the treasure in some node of the underlying graph  $\mathcal{G}$. The agent  has no {\em a priori} knowledge of the graph, of the location of the treasure or of the initial distance to it, and has to find the treasure.  {The \emph{cost}} of a treasure hunt algorithm is the worst-case number of edge traversals performed by the agent until finding the treasure. 

In order to state our problem we need the notion of a {\em ball}.
Given a non-negative integer $k$, a graph $G$, and a node $u$, the ball $B_k(G,u)$ is defined as the subgraph $K=(V_K,E_K)$ of $G$, where $V_K$ is the set of all nodes at distance at most $k$ from $u$ in $G$, and $E_K$ is the set of all edges of $G$ whose at least one extremity is at distance smaller than $k$ from $u$ in $G$. (Thus  $B_k(G,u)$ is the subgraph of $G$ induced by nodes at distance at most $k$ without edges joining nodes at distance exactly $k$ from $u$ in $G$). The number of edges in ball $B_k(\mathcal{G},s)$, where $s$ is the source node, will be denoted by $e(k,\mathcal{G})$. Whenever the graph $\mathcal{G}$ is clear from the context, we will write $e(k)$ instead of $e(k,\mathcal{G})$.



The main problem considered in this paper is inspired by the following conjecture of Awerbuch, Betke, Rivest and Singh \cite{ABRS}, formulated for their fuel-restricted model:

\begin{quotation}

Is it possible (we conjecture not) to find a treasure in time nearly linear in the number of those vertices and edges whose distance to the source is less than or equal to that of the treasure?
\footnote{Time in this conjecture is what we call cost, i.e., the worst-case number of edge traversals until finding the treasure.}

\end{quotation}

\subsection{Our results} \label{results}

Our main result refutes the above twenty-year-old conjecture. Let $d$ be any integer such that $1<d\leq r$, where $r$ is the radius of the underlying graph $\cal G$.
We first design a deterministic treasure hunt algorithm working in the unrestricted model and always finding a treasure located at distance at most $d$ from the source node, at cost $\mathcal{O}(e(d) \log d)$. We then show how to modify this algorithm to work in the fuel-restricted and rope-restricted models with the same complexity. 
 Since $d \leq e(d)$, the cost of our algorithms differs from $e(d)$ only by a logarithmic factor, and hence it is nearly linear in $e(d)$, contrary to the conjecture. 
 Due to the ignorance of the agent concerning the graph in which it operates, no treasure hunt algorithm can beat cost $\Theta(e(d))$ for all graphs (cf. Proposition \ref{lb}) and thus 
 our algorithms are also almost optimal. The main difficulty is to design the algorithm for the unrestricted model. This algorithm is then suitably modified for each of the two restricted models.

Solving the problem of treasure hunt at a cost quasi-linear in $e(d)$ required to respect two fundamental principles, whose joint implementation seemed precarious in the light of the existing literature.

The first one is a \emph{prudence principle}. It consists in never getting ``for too long'' beyond the unknown distance $d$ in order to guarantee a cost that depends on $e(d)$. This can be ideally achieved by  emulating BFS. However, since in such an emulation the agent must physically move from one node to the next, it may be forced to traverse $\Omega(e(d)^2)$ edges before finding the treasure, in some graphs. In particular, this could be the case when $\mathcal{G}$ is an infinite line.

The second principle is what we could call an \emph{efficiency principle}. It consists in getting a cost that is asymptotically close to the number of edges of the subgraph that has been explored till finding the treasure, if the treasure is far away. This can be ideally achieved using the treasure hunt algorithm of \cite{DKK}, the cost of which is linear in the number of edges of the explored subgraph. However, using this algorithm, the agent may go for too long beyond the unknown distance $d$ and consequently the cost of treasure hunt could not be upper bounded by any function of $e(d)$.
The key challenge overcome by our work was combining these two principles within the same algorithm. It is precisely the combination of prudence with efficiency that finally made possible the design of an almost-optimal treasure hunt algorithm.

\subsection{Related work}

The task of treasure hunt, i.e., finding an inert target hidden in some environment, has been studied for over fifty years \cite{Bec1964,BN1970,Bel1963}.
The environment where the target is hidden may be a graph or a plane, and the search may be deterministic or randomized.
The book \cite{AG} surveys both treasure hunt and the related rendezvous problem, where the target and the searching agent are both mobile and
they cooperate to meet. This book is concerned mostly with randomized search strategies. In \cite{MP,TSZ} the authors studied relations between treasure hunt and rendezvous in graphs.  The authors of \cite{BCR} studied the task of treasure hunt on the line and in the grid, and initiated the study of the task
of searching for an unknown line in the plane. This research was continued, e.g., in \cite{JL,La2}. 

Several papers considered  treasure hunt in the plane, see surveys \cite{Gal2013,GK2010}.
In  \cite{La}, the author designs an optimal algorithm to sweep the plane in order to locate an unknown fixed target, where locating means getting the agent originating at point $O$ to a point $P$ such that the target is in the segment $OP$. In \cite{FHGTM}, the authors generalized the search problem in the plane to the case of several searchers.  Efficient treasure hunt in the plane, under complete ignorance of the searching agent, was studied in \cite{Pe}.
Treasure hunt on the line (called the cow-path problem \cite{KRT}) has been also generalized to the environment consisting of multiple rays originating at a single point \cite{AAD,DFG2006,LS2001,Sch2001}. 

In \cite{FKK+2008}, the authors considered treasure hunt in several classes of graphs including trees. Treasure hunt in trees was studied in \cite{DCD,DCS1995,KZ2011}. In \cite{DCD,DCS1995}, the authors considered complete $b$-ary trees, and in \cite{KZ2011}, treasure hunt was studied in symmetric trees, with possibly multiple treasures.

In \cite{KKKS,MP}, treasure hunt in graphs was considered under the advice paradigm, where a given number of bits of advice can be given to the agent, and the issue is to minimize this number of bits.  The impact of different types of knowledge on the efficiency of the treasure hunt problem restricted to symmetric trees was studied in \cite{KZ2011}.

The two papers closest to the present work are \cite{ABRS,DKK}. Both of them are mainly interested in exploration of finite unknown graphs but they both get interesting corollaries for the treasure hunt problem. \cite{ABRS} adopts the fuel-restricted model and \cite{DKK} adopts the rope-restricted model. In \cite{ABRS},
 the authors get a treasure hunt algorithm working at cost $O(E+V^{1+o(1)})$, where $E$ (resp. $V$) is the number of edges (resp. nodes) in a ball $B_{\Delta}(\mathcal{G},s)$, with $\Delta \leq d +o(d)$,
 if the treasure is at distance at most $d$ from the starting node of the agent.
 Since $e(\Delta)$ may be a lot larger than $e(d)$, this does not permit to bound the cost of the algorithm by any function of $e(d)$. This impossibility may be the reason for their conjecture that we refute in this paper. In \cite{DKK}, the authors design, for any constant $0<\alpha <1$, a treasure hunt algorithm whose cost is linear in $e((1+\alpha)d)$. Again, since $e((1+\alpha)d)$ may be much larger than $e(d)$, this does not permit to bound the cost of the algorithm by any function of $e(d)$.

\section{Preliminaries}

In this section we introduce some conventions, definitions and procedures that will be used to describe and analyze our algorithm.


Consider any graph $G=(V_G,E_G)\subseteq \mathcal{G}$. If $G$ is finite, its size \ie its number of edges is denoted by $|G|$. A graph is said to be {\em empty} if it contains no node. In the rest of this section, we assume that $G$ is not empty.


Let $u$ and $v$ be two (not necessarily distinct) nodes of $G$. We say that a sequence of $i$ integers $(x_1,x_2,\ldots,x_i)$ is a {\em path} (of length $i$) in $G$ from node $u$ to $v$ iff (1) $i=0$ and $u=v$, or (2) there exists an edge $e$ in $G$ between node $u$ and a node $w$ of $G$ such that the port number of edge $e$ at node $u$ is $x_1$ and $(x_2,\ldots,x_i)$ is a path from node $w$ to $v$ in $G$. The lexicocraphically smallest shortest path from node $u$ to $v$ in $G$, if any, is denoted by $P_G(u,v)$, and the length of this path is denoted by $|P_G(u,v)|$. The distance between $u$ and $v$ in $G$ is denoted by $d_G(u,v)$ and is equal to $|P_G(u,v)|$ if $P_G(u,v)$ exists, $\infty$ otherwise. If $G$ is finite and connected, {the eccentricity $\epsilon_G(u)$ of node $u$} is defined as $\max_{w\in{V_G}} d_G(u,w)$. The degree of $u$ in $G$ will be denoted by $deg_{G}(u)$, or simply by $deg(u)$ if $G=\mathcal{G}$. We say that node $u$ is \emph{incomplete} (resp. \emph{complete}) in $G$ if $deg_G(u)<deg(u)$ (resp. $deg_G(u)=deg(u)$). We also say that a port $p$ is free at node $u$ in $G$, if $p\leq deg(u)-1$ and there is no edge $(u,*,p,*)$ or $(*,u,*,p)$ in $E_G$.


We will often need to handle subgraphs of $\mathcal{G}$ through union and intersection operations. More precisely, given two subgraphs $G'=(V_{G'},E_{G'})$ and $G''=(V_{G''},E_{G''})$ of $\mathcal{G}$, the union of (resp. the intersection of) $G'$ and $G''$ is denoted by $G' \sqcup G''$ (resp. $G' \sqcap G''$) and is equal to $(V_{G'}\cup V_{G''},E_{G'}\cup E_{G''})$ (resp. $(V_{G'}\cap V_{G''},E_{G'}\cap E_{G''})$).

We define the boundary of a ball $B_{f}(\mathcal{G},s)$, where $s$ is the source node, as the set of nodes $u$ satisfying the following condition: $u$ is a node of $B_{f}(\mathcal{G},s)$ and for each neighbor $v$ of $u$ in $B_{f}(\mathcal{G},s)$, $d_{B_{f}(\mathcal{G},s)}(s,v)\leq d_{B_{f}(\mathcal{G},s)}(s,u)$.

To design our algorithm, we will also make use of three basic routines presented below.
The first routine is ${\tt MoveTo}(G,v)$. Assuming that the agent currently occupies a node $w$ of $G$ and $P_G(w,v)$ exists, this routine moves the agent from node $w$ to node $v$ by following path $P_G(w,v)$. The second routine is ${\tt IncompleteNodes}(v,G,l)$ where $l$ is a positive integer. This routine returns the set of all nodes $w$ of $G$ such that $d_G(v,w)\leq l$ and $w$ is incomplete in $G$. The third routine is ${\tt Nodes}(\mathcal{S})$, where $\mathcal{S}$ is a finite set of finite subgraphs of $\mathcal{G}$. This routine returns the union of all nodes in all subgraphs from $\mathcal{S}$.

Given an execution $\mathcal{E}$ of a series of instructions, the cost of $\mathcal{E}$ is the number of edge traversals performed by the agent during $\mathcal{E}$.

We will use the following convention.
The agent will sometimes need to use Depth First Search traversal of graphs (not performed physically, but performed as a computation in the memory of the agent). Such a traversal depends on the order in which edges incident to a given node are traversed for the first time. We fix this order as the increasing order of port numbers at the given node. In this way the traversal is unambiguous, and we call it DFS.


We end this section with the following straightforward observation implying that no treasure hunt algorithm can beat cost $\Theta(e(d))$ for all graphs and hence our treasure hunt algorithm is almost optimal
\footnote{Since the treasure has to be hidden in a node, the agent does not necessarily have to traverse all edges of the ball $B_d(\mathcal{G},s)$.}.
This observation holds in all three considered models: unrestricted, fuel-restricted and rope-restricted. In fact, the proposition shows that graphs $\mathcal{G}$ for which the cost of treasure hunt is at least $\Theta(e(d))$ can be found for any density of edges in the ball $B_d(\mathcal{G},s)$.

\begin{proposition}\label{lb}
For every treasure hunt algorithm $\cal A$, for every integer $d>1$, for every integer $m\geq d$ and for every integer {$m\leq x \leq m^2$}, there exists a graph $\mathcal{G}$ {of radius $d$} such that $B_d(\mathcal{G},s)$ has $\Theta(m)$ nodes and $\Theta(x)$ edges and the cost incurred by $\cal A$ to find the treasure located at some node at distance at most $d$ from the source node $s$ in graph $\mathcal{G}$ is at least $e(d,\mathcal{G})-1$.
\end{proposition}

\begin{proof}
It is enough to prove the proposition for the unrestricted model. Our proof works even if the radius $d$ is known to the agent. Fix a treasure hunt algorithm $\cal A$, an integer $d>1$, an integer $m\geq d$ and an integer {$m\leq x \leq m^2$}.

We start with the construction of the following graph $\mathcal{H}$. The set of nodes of  $\mathcal{H}$ is the union of disjoint sets $\{s\}\cup A\cup B$, where $A=\{a_1,a_2,\dots,a_d\}$ and $B=\{b_1,b_2,\dots,b_m\}$. The set of edges of $\mathcal{H}$ is the union of disjoint sets $X \cup Y \cup Z$, where $X=\{\{s,a_1\},\{a_1,a_2\},\{a_2,a_3\},\dots, \{a_{d-1},a_d\}\}$, $Y=\{\{s,b_1\}, \{s,b_2\},\dots,\{s,b_m\}\}$, and $Z$ is any set of edges of size $\min(x,m(m-1)/2)$ between nodes of the set $B$. The graph $\mathcal{H}$ is connected, has $1+d+m \in \Theta(m)$ nodes and $d+m+\min(x,m(m-1)/2) \in \Theta(x)$ edges. Moreover, $B_d(\mathcal{H},s)=\mathcal{H}$ and $\mathcal{H}$ has radius $d$.

Let $x$  be the first step when all nodes of $\mathcal{H}$ are visited, assuming that the treasure has not  been found before. Let $y$ be the number of untraversed edges of $\mathcal{H}$ at step $x$.  There are two cases. If $y<2$ then the adversary puts the treasure at the last-discovered node of $\mathcal{H}$ and the proposition is satisfied by graph $\mathcal{H}$ itself. Hence we may assume that $y\geq 2$.
Let $e_1$, ..., $e_y$ be the edges of $\mathcal{H}$ untraversed by step $x$, and let $\mathcal{H}_i$, for $i=1,\dots, y$, be the graph $\mathcal{H}$ with a midpoint added on edge $e_i$. Note that none of the edges $e_i$ can be an element of the set $X$ because this would contradict the fact that node $a_d$ has been visited by step $x$. Consequently, the radius of each graph $\mathcal{H}_i$ is $d$. 
The set of nodes of $B_d(\mathcal{H}_i,s)$ is the set of nodes of $B_d(\mathcal{H},s)$ augmented by the midpoint added on edge $e_i$, and we have $e(d,\mathcal{H}_i) = e(d,\mathcal{H})+1$.
Thus, for all $i$, $B_d(\mathcal{H}_i,s)$ has $\Theta(m)$ nodes and $\Theta(x)$ edges.  In this case the adversary will put the treasure at one of the added midpoints and claim that the actual graph is the corresponding graph $\mathcal{H}_i$.
Below we show how to choose the index $i$. Since the treasure is not placed at any node of $\mathcal{H}$, the execution of $\cal A$ until step $x$ is the same in graphs $\mathcal{H}$, $\mathcal{H}_1$,...,$\mathcal{H}_y$. In order to find the treasure placed at the added node of some of the graphs $\mathcal{H}_i$, algorithm $\cal A$ must take the port corresponding to edge $e_i$ at one of its extremities. If it does not find a midpoint inserted in $e_i$, i.e., if the actual graph is not $\mathcal{H}_i$, algorithm $\cal A$ must  take the port corresponding to some other edge $e_j$ at one of its extremities, in order to find the treasure situated at the added midpoint of $e_j$, if the actual graph were $\mathcal{H}_j$, and so on. Let $e_f$ be the last of the edges $e_i$ examined in this way. Suppose that the actual graph $\mathcal{G}$ is $\mathcal{H}_f$ and that the adversary places the treasure at the midpoint of $e_f$. Hence the cost of finding the treasure in the graph $\mathcal{G}=\mathcal{H}_f$ is {at least} $e(d,\mathcal{H})-y+y=e(d,\mathcal{H})\geq e(d,\mathcal{G})-1$. This concludes the proof.
 \end{proof}

\section{Intuition} 
\label{sec:int}
The purpose of this section is to sketch an intuitive overview of our algorithm that allows to find the treasure at an almost-optimal cost in the unrestricted model. To this end and to simplify the discussion, we will assume that the underlying graph $\mathcal{G}$ is countably infinite with nodes of finite degrees. We will rely on the notion of \emph{largest explored ball}. By ``largest explored ball'', at a given phase of treasure hunt, we mean the ball $B_f(\mathcal{G},s)$ where $f$ is the largest integer such that each edge of $B_f(\mathcal{G},s)$ has been traversed at least once. This largest integer $f$ is the radius of the largest explored ball.



At a high level, our algorithm works in phases $i=1,2,3,\ldots$ and immediately stops as soon as the treasure is found. At the beginning of phase $i$, the agent is located at node $s$ and {the radius of the largest explored ball is equal to $f_i$}. The goal for the agent is to terminate the phase at node $s$ while satisfying at least one of the following three conditions unless, of course, the treasure has been found before.

\begin{itemize}
\item {\it Condition~1.} The agent has entirely explored ball $B_{f_i+1}(\mathcal{G},s)$, $e(f_i+1)\geq 2e(f_i)$ and the cost of the phase is $\mathcal{O}(e(f_i+1))$. 
\item {\it Condition~2.} The agent has entirely explored ball $B_{2f_i}(\mathcal{G},s)$, $f_i\geq 1$ and the cost of the phase is $\mathcal{O}(e(f_i))$.
\item {\it Condition~3.} The agent has entirely explored ball $B_{f_i+k}(\mathcal{G},s)$ for some positive integer $k$, $e(f_i+k+1)\geq 2e(f_i)$, $f_i\geq 2$, and the cost of the phase is $\mathcal{O}(e(f_i)\log f_i)$.
\end{itemize}

Actually, the conditions we really seek to meet in our algorithm are a little more intricate than those presented above, because we needed stronger technical requirements to show Theorem~\ref{theo:final}, which refutes the conjecture of Awerbuch, Betke, Rivest and Singh \cite{ABRS}. However, this would add an unnecessary level of complexity to understand the intuition, hence we omit these technical details here.

Before seeing how we implement our strategy, let us briefly examine why it permits us to get a cost quasi-linear in $e(d)$. Since $f_1=0$ and the radius of the largest explored ball increases by at least one during each phase in which the treasure is not found, the agent necessarily finds the treasure by the end of some phase $\lambda\leq d$, and $f_i<f_\lambda<d$ for every $1\leq i < \lambda$. During each phase satisfying Condition~1, the size of the largest explored ball at least doubles, which means that the total cost of these phases is upper bounded by twice the worst-case cost of the last phase satisfying Condition~1 i.e., $\mathcal{O}(e(f_\lambda+1))$. Concerning the phases fulfilling Condition~2, their number is at most {$\mathcal{O}(\log (f_{\lambda}+1))$} and the cost of each of them cannot be more than $\mathcal{O}(e(f_\lambda))$, which implies that their total cost is {$\mathcal{O}(e(f_\lambda)\log (f_{\lambda}+1))$}. It remains to consider the case of the phases satisfying Condition~3. Given such a phase $i$, we have the guarantee that the size of the largest explored ball at least doubles between the beginning of phase $i$ and the end of phase $i+1$, provided phase $i+1$ exists and is not prematurely interrupted by the discovery of the treasure. Indeed, at the end of phase $i$, the agent has at least entirely explored ball $B_{f_i+k}(\mathcal{G},s)$ for some positive integer $k$ and $e(f_i+k+1)\geq 2e(f_i)$, while at the end of the (not prematurely interrupted) phase $i+1$ the agent has at least entirely explored ball $B_{f_{i+1}+1}(\mathcal{G},s)$ with $f_{i+1}\geq f_i+k$. Using this, it can be shown that the total cost of the phases satisfying Condition~3 is at most four times the worst-case cost of the last phase satisfying this condition i.e., {$\mathcal{O}(e(f_\lambda)\log (f_{\lambda}+1))$}. Given that the last phase $\lambda$ can be viewed as a truncated phase that should have normally satisfied one of the three conditions, our sketch of analysis leads to the conclusion that the cost incurred by the agent till the discovery of the treasure is in {$\mathcal{O}(e(f_\lambda+1)\log (f_{\lambda}+1))$}, which is $\mathcal{O}(e(d)\log d)$ and is in line with our expectations.

Having justified the pertinence of such a strategy, we can turn our attention to its implementation. To do so, we need to introduce a technical building block, called {\tt GlobalExpansion}$(l,m)$ to which we will go back at the end of this section to give additional details. Always executed from the source node $s$, it is a function that returns a boolean and whose two input parameters are positive integers except $m$ that may be sometimes equal to the special symbol $\perp$. Assuming that $B_{f}(\mathcal{G},s)$ is the largest explored ball, the execution of {\tt GlobalExpansion}$(l,\perp)$ permits the agent to traverse all the edges of $B_{f+l}(\mathcal{G},s)$ that are outside of $B_{f}(\mathcal{G},s)$ before coming back to node $s$. Under the same assumption, the execution of {\tt GlobalExpansion}$(l,m)$, when $m$ is a positive integer, consists for the agent in acting as if $m$ was $\perp$ but with the following extra requirement: as soon as more than $m$ distinct edges outside of $B_{f}(\mathcal{G},s)$ have been traversed during the execution of the function, the agent backtracks to node $s$ and aborts this execution. If $m$ is $\perp$ or at least large enough to avoid an aborted execution, the agent ends up exploring $B_{f+l}(\mathcal{G},s)$ and the function returns true. Otherwise, the function returns false. It should be stressed that all of this is made while guaranteeing two properties. The first one is that the agent is always in $B_{f+2l-1}(\mathcal{G},s)$ during the execution of {\tt GlobalExpansion}$(l,m)$. The second is that the cost of the execution of {\tt GlobalExpansion}$(l,m)$ is $\mathcal{O}(e(f+2l-1))$ (resp. $\mathcal{O}(\min\{e(f)+m,e(f+2l-1)\})$) when $m=\perp$ (resp. $m\ne\perp$). Both these properties will turn out to be crucial to ensure a proper design of the phases. Finally, even if by chance the agent could explore a larger ball, we will assume for the ease of our intuitive explanations that $B_{f+l}(\mathcal{G},s)$ (resp. $B_{f}(\mathcal{G},s)$) is the largest ball explored by the agent at the end of {\tt GlobalExpansion}$(l,m)$ in the case where the returned value is true (resp. false).


Let us consider a phase $i$ of our algorithm and, in order not to burden the text with a lot of ``unless the treasure is found'', let us assume that the treasure will not be found by the end of it. Phase $i$ is made of at most three successive attempts, each of them aiming at fulfilling  at least one of the three conditions described earlier, with the help of our building block. In the first attempt, the agent executes {\tt GlobalExpansion}$(1,\perp)$ from node $s$, the cost of which is $\mathcal{O}(e(f_i+1))$. At the end of this execution, the agent is at node $s$ and $B_{f_i+1}(\mathcal{G},s)$ has been entirely explored by the agent. If $e(f_i+1)\geq 2e(f_i)$ or $f_i\leq 1$, the first attempt is a success as Condition~1 or Condition~2 is verified, and the agent directly switches to phase $i+1$. Otherwise, the attempt is a failure, but we can nonetheless observe that the cost incurred because of the attempt is just $\mathcal{O}(e(f_i))$ because $e(f_i+1)<2e(f_i)$.

If the first attempt has failed, the agent starts the second attempt of phase $i$ that consists of an execution of function {\tt GlobalExpansion}$(f_i-1,e(f_i))$. The hope here is to expand by a distance of $f_i-1$ the radius of the largest explored ball, which is $B_{f_i+1}(\mathcal{G},s)$. According to the properties of {\tt GlobalExpansion} and the fact that $e(f_i+1)<2e(f_i)$, the cost of this execution, and thus of the second attempt, is $\mathcal{O}(e(f_i))$. If {\tt GlobalExpansion}$(f_i-1,e(f_i))$ returns true, then at the end of the second attempt, the radius of the largest explored ball is $2f_i$. Hence, the cost of the first two attempts being equal to $\mathcal{O}(e(f_i))$ and $f_i$ being at least $2$, Condition~2 is satisfied and the agent starts phase $i+1$ without making the third attempt.


On the other hand, if {\tt GlobalExpansion}$(f_i-1,e(f_i))$ returns false, it is a different story. Indeed, the largest explored ball is still only $B_{f_i+1}(\mathcal{G},s)$ and we cannot ensure the fulfillment of Condition~1 or Condition~2. This is exactly where Condition~3 comes into the picture. In order to remedy the failures of the two previous attempts, the agent will start a third and last attempt which consists of a dichotomic process that is described in Algorithm~\ref{alg:dich}. At the end of this process, Condition~3 is guaranteed to be satisfied.

\begin{algorithm}[H]
\caption{Third attempt\label{alg:dich}}
\SetKwFunction{GE}{GlobalExpansion}
$floor := f_i+1$; $ceil := 3f_i-2$; $l := \lfloor \frac{ceil-floor}{2} \rfloor$\;\label{line1}
\While{$l\geq 1$ \AND $|B_{floor}(\mathcal{G},s)|< 2e(f_i)$}{
$success :=$ \GE{$l, e(f_i)$}\;
      \eIf{$success = \true$ }{
      	$floor := floor + l$; $l := \lfloor \frac{ceil-floor}{2} \rfloor$\;
      }
      {
        $ceil := floor + 2l -1$; $l := \lfloor\frac{l}{2} \rfloor$\;
      }
  }
\end{algorithm}

In order to better understand why we can get such a guarantee, let us take a look at the properties that are satisfied during the third attempt and at its end.

Since the execution of {\tt GlobalExpansion}$(f_i-1,e(f_i))$ returned false, the agent has explored at least $e(f_i)$ distinct edges outside of ball $B_{f_i+1}(\mathcal{G},s)$ during the second attempt. Moreover, during this execution, the agent was always in $B_{3f_i-2}(\mathcal{G},s)$ according to the properties of {\tt GlobalExpansion}. As a result, in view of line~\ref{line1} of Algorithm~\ref{alg:dich}, we necessarily have the following feature before the execution of the while loop of Algorithm~\ref{alg:dich}: $B_{floor}(\mathcal{G},s)$ is the largest explored ball and $e(ceil)\geq 2e(f_i)$. Actually, by carefully examining the pseudocode of the while loop and using again the properties of {\tt GlobalExpansion}, it can be inductively proven that this feature is a loop invariant. Alone, this loop invariant is not enough to bring the sought guarantee, but as highlighted below, it is of precious help to do the job.

The number of iterations of the while loop can be shown to be $\mathcal{O}(\log f_i)$. Furthermore, at the beginning of each iteration, $B_{floor}(\mathcal{G},s)$ has size smaller than $2e(f_i)$ in view of the condition of the while loop, and is the largest explored ball in view of the loop invariant. Hence, according to the cost property of {\tt GlobalExpansion}, each execution of {\tt GlobalExpansion}$(l,e(f_i))$ costs at most $\mathcal{O}(e(f_i))$ like the previous two attempts, which gives a total cost of $\mathcal{O}(e(f_i)\log f_i)$ of the whole phase. This corresponds exactly to the target value of Condition~3. Along with this, at the end of the while loop, the size of $B_{floor}(\mathcal{G},s)$ is at least $2e(f_i)$, or $l<1$. In the first case, we immediately have $e(floor+1)\geq 2e(f_i)$, while in the second case it can be shown that $ceil\leq floor+1$. This, combined with the fact that $e(ceil)$ is always at least $2e(f_i)$ (by the loop invariant) and the fact that $floor$ is always at least $f_i+1$, allows us to show the last missing piece of the puzzle, which is precisely this: when Algorithm~\ref{alg:dich} terminates, ball $B_{f_i+k}(\mathcal{G},s)$ is entirely explored and $e(f_i+k+1)\geq 2e(f_i)$ for some integer $k\geq 1$.

To conclude with the intuitive explanations, let us give, as promised, some more insight concerning the building block {\tt GlobalExpansion}$(l,m)$. At first glance, one might think that {\tt GlobalExpansion} could be directly derived from the exploration algorithm ${\tt CFX}(v,r,\alpha)$ of \cite{DKK}, which permits to explore a ball $B_r(\mathcal{G},v)$ at a cost of {$\mathcal{O}\left(\frac{|B_{(1+\alpha)r}(\mathcal{G},v)|}{\alpha}\right)$} for any given real $\alpha>0$ (this corresponds to a cost of {$\mathcal{O}\left(\frac{e((1+\alpha)r)}{\alpha}\right)$} when $v=s$) {provided $\alpha r\geq 1$}. Indeed, the task of {\tt GlobalExpansion}$(l,m)$ that consists in expanding the radius $f$ of the largest explored ball by a distance $l$ in the case where $m$ is appropriately set, can be done with ${\tt CFX}(s,f+l,\alpha)$. However, in this case we want the cost of this expansion to be $\mathcal{O}(e(f+2l-1))$, which is an important property of our strategy. This cannot be guaranteed using ${\tt CFX}(s,f+l,\alpha)$ because, in order to get a cost depending on $e(f+2l-1)$, we would have to set $\alpha$ to a value lower than $\frac{l-1}{f+l}$, which {cannot lead to a cost that is linear in $e(f+2l-1)$}, as $\frac{l-1}{f+l}$ can be arbitrarily small. True, during the design we could have been ``less demanding'' about some of the properties of {\tt GlobalExpansion}$(l,m)$, but not significantly enough to permit the use of ${\tt CFX}(s,f+l,\alpha)$ without spoiling the validity or the cost complexity of our strategy. Another solution that may come to mind would be to apply ${\tt CFX}(v,l,\alpha)$ from each node $v$ located on the boundary of the largest explored ball $B_f(\mathcal{G},s)$. Visiting each node of the boundary can be done in $\mathcal{O}(e(f))$. Hence, this solution looks attractive because by setting $\alpha$ to $\frac{1}{2}$ or less (which overcomes the above problem of the arbitrarily small value) and provided the zones explored by the different executions of {\tt CFX} do not overlap, we would get a cost that is linear in $e(f+2l-1)$. The bad news is that there may be overlaps. Of course, some overlaps can be easily avoided, especially those appearing within $B_f(\mathcal{G},v)$, but some others cannot without running the risk of missing some nodes of $B_{f+l}(\mathcal{G},s)$ that are outside of $B_{f}(\mathcal{G},s)$. These ``necessary overlaps'' may be pernicious and may occur in a way that prevents us from guaranteeing a cost of $\mathcal{O}(e(f+2l-1))$.

So, what did we do? Although it was not possible to use ${\tt CFX}$ as a black box, we managed to tailor {\tt Global\-Expansion} by adapting to our needs an elegant algorithmic technique used in ${\tt CFX}$. Through a set of judiciously pruned trees spanning some already explored area, it allowed us to satisfy the desired cost property of {\tt Global\-Expansion} by controlling and amortizing efficiently the number of times the same edges are traversed. The technique in question is detailed in the next section that presents the pseudocode of our treasure hunt algorithm.

\section{Algorithm}\label{sec:alg}

Solving the treasure hunt problem  in the unrestricted model can be done by executing Algorithm ${\tt TreasureHunt}(x)$ described below in Algorithm~\ref{alg:TH} and by interrupting it as soon as the treasure is found. The input parameter $x$ is a positive real constant. It is a technical ingredient that will have an impact on the maximal distance at which the agent can be from node $s$. In our present context, parameter $x$ does not really matter and it can be fixed as \emph{any} positive real constant. In fact, it will show its full significance in Section~\ref{sec:rest} that is dedicated to the same problem in restricted models: there, we will reuse ${\tt TreasureHunt}(x)$ in a context where $x$ will have to be carefully chosen. The variable $\mathcal{M}$ in line~\ref{THadd} of Algorithm~\ref{alg:TH} is a global variable that will always correspond to some explored subgraph of $\mathcal{G}$. For this reason, it will recurrently appear in most of the pseudocodes of the functions described thereafter.

\begin{algorithm}[H]
\caption{${\tt TreasureHunt}(x)$\label{alg:TH}}
\SetKwFunction{GE}{GlobalExpansion}
\SetKwFunction{SE}{Search}

\SetKwBlock{Repeat}{repeat}{}
$v :=$ the current node\; \label{TH1}
$\mathcal{M} := (\{v\},\emptyset)$; \tcc{$\mathcal{M}$ is a global variable}\label{THadd} 
\Repeat{
       \SE{$x$}\;
  }
\end{algorithm}

As the reader can see, the execution of Algorithm ${\tt TreasureHunt}(x)$ essentially consists of a series of executions of procedure {\tt Search}$(x)$, whose pseudocode is described in Algorithm~\ref{alg:SE}: these executions correspond to what we called ``phases'' in our intuitive explanations of Section~\ref{sec:int}. Procedure {\tt Search}$(x)$ should be seen as the organizer of our solution. At the beginning of each call to {\tt Search}$(x)$, $\mathcal{M}$ is some explored ball $B_f(\mathcal{G},s)$ and the goal of the call is to make grow this ball while satisfying some conditions. These conditions, whose simplified version we gave at the beginning of Section~\ref{sec:int}, are formally described in Lemma~\ref{lem:search1}.

\begin{algorithm}[H]
\caption{${\tt Search}(x)$\label{alg:SE}}
\SetKwFunction{LE}{LocalExpansion}
\SetKwFunction{MT}{MoveTo}
\SetKwFunction{CDFS}{CDFS}
\SetKwFunction{BOUND}{Boundary}
\SetKwFunction{MAP}{Map}
\SetKwRepeat{Do}{do}{while}
$v :=$ the current node; $m := |\mathcal{M}|$\;\label{SEf}
$floor := \epsilon_\mathcal{M}(v)$; $ceil := \lfloor (1+x)\cdot floor \rfloor$\;\label{assign0}
$success :=$ \GE{$1, \perp$}\;\label{SEinit}
$floor := floor+1$; $i := 0$; $l := \lfloor \frac{ceil-floor}{2} \rfloor$\; \label{assign1}
\While{$l\geq 1$ \AND $|\mathcal{M}|< 2m$ \AND $(i\ne 1$ \OR $success= \false)$}{\label{SEcondwhile}
$success :=$ \GE{$l, m$}\;
      \eIf{$success = \true$ }{\label{condassign2}
      	$floor := floor + l$; $l := \lfloor \frac{ceil-floor}{2} \rfloor$\;\label{assign2}
      }
      {
        $ceil := floor + 2l -1$; $l := \lfloor\frac{l}{2} \rfloor$\;\label{assign3}
      }\label{fincond}
$\mathcal{M} := B_{floor}(\mathcal{M},v)$\;\label{resizeball}
$i := i+1$\;
  }
\end{algorithm}

Although there are some technical differences, we can discern, throughout the lines of Algorithm~\ref{alg:SE}, the three attempts outlined in Section~\ref{sec:int} that rely on function {\tt GlobalExpansion}. Roughly speaking, line~\ref{SEinit} of Algorithm~\ref{alg:SE} relates to the first attempt, the first iteration of the while loop of Algorithm~\ref{alg:SE} relates to the second attempt, and the other iterations relate to the third attempt. 

The pseudocode of function ${\tt GlobalExpansion}(l, m)$ is given by Algorithm~\ref{alg:GE}. It has primarily the same specifications as those given in Section~\ref{sec:int} except that we did not implement the case where $m=\perp$ and $l\geq2$ as it was not necessary for our purpose. Hence, the function precisely handles the case where $l=1$ and $m=\perp$, and the case where $l\geq1$ and $m\ne\perp$. The general scheme of the function is as follows. At the beginning, the agent knows a ball $B_{f}(\mathcal{G},s)$ that is stored in variable $\mathcal{M}$ and the objective is to expand the radius of this ball by a distance $l$, without exploring more than $m$ edges outside of $B_{f}(\mathcal{G},s)$, if $m\ne\perp$. To do this, the agent visits the nodes $L[1],L[2],\ldots$ (stored in the array $L$) of the boundary of $B_f(\mathcal{G},s)$ and executes from these nodes function {\tt CDFS} (described in Algorithm~\ref{alg:CDFS}) or function {\tt LocalExpansion} (described in Algorithm~\ref{alg:LE}) depending on the initial values of $l$ and $m$. Each of these executions, which starts and ends at the same node, locally contributes to the global expansion of the ball. In the case where $m\ne\perp$, variable $b$ of Algorithm~\ref{alg:GE} is updated with the return value of the two aforementioned functions, and corresponds at each stage to the remaining number of new edges the agent is authorized to traverse outside of $B_{f}(\mathcal{G},s)$. If $b$ becomes negative before the end of the while loop of Algorithm~\ref{alg:GE}, the objective of expansion is simply not reached. Note that, in order to avoid that the moves from one node of the boundary of $B_f(\mathcal{G},s)$ to the next get too costly, they are made according to a precise order that results from the definition of $L$ given in line~\ref{DFS} of Algorithm~\ref{alg:GE}.

\begin{algorithm}[H]
\caption{${\tt GlobalExpansion}(l, m)$\label{alg:GE}}
\SetKwFunction{LE}{LocalExpansion}
\SetKwFunction{MT}{MoveTo}
\SetKwFunction{CDFS}{CDFS}
\SetKwFunction{BOUND}{Boundary}
\SetKwFunction{MAP}{Map}
$v := $ the current node\;\label{GEl1}
$L :=$ the array containing all the nodes of the boundary of $\mathcal{M}$ sorted in the order of the first visit through the DFS traversal of $\mathcal{M}$ from node $v$\;\label{DFS}
$T :=$ the tree produced by the DFS traversal of $\mathcal{M}$ from node $v$\;\label{DFStree}
$i := 1$; $b := m$; $\mathcal{T} := \emptyset$; \tcc{$\mathcal{T}$ is a global variable}\label{GE:initialize}
\While{$i\leq |L|$ \AND $(b\geq0$ \OR $b=\perp)$}{\label{GE:while}
 \MT{$T, L[i]$}\;\label{GElMT}\label{firsttype}
 \eIf{$l=1$}{\label{secondtype0}
\eIf{$b=\perp$}{
\tcc{We run \CDFS{$1,deg(L[i])$} without using its return value.}
$(*,*) := $\CDFS{$1,deg(L[i])$}\;\label{secondtype1}
}
{
\tcc{We run \CDFS{$1,b$} without using the second term of its return value.}
$(b,*) := $\CDFS{$1,b$}\;\label{secondtype1bis}
}
}
{
$b :=$ \LE{$l, b$}\; \label{secondtype2}
}
$i := i+1$\;
  }
\MT{$T,v$}\;\label{GEback}\label{thirdtype}
{\Return the logical value of ``$b\geq0$ or $b=\perp$''\;}\label{GE:lastline}
\end{algorithm}

As one can see in lines~\ref{secondtype1} and~\ref{secondtype1bis} of Algorithm~\ref{alg:GE}, the implementation of the case $l=1$ in Algorithm~\ref{alg:GE} directly relies on function {\tt CDFS}. We will see below that this function is also involved in the trickier case where $l\geq2$ and $m\ne\perp$ through the calls to function {\tt LocalExpansion}. Function {\tt CDFS}$(l,b)$ permits the agent to perform a depth-first search in the zone that does not belong to $\mathcal{M}$ when it starts executing it. During the execution of this function $\mathcal{M}$ grows, augmented with the edges that are traversed by the agent. The two input parameters $l\geq 1$ and $b\geq0$ are integers that bring constraints to the execution of the depth-first search. The first indicates the limit depth of the search, while the second indicates an upper bound on the number of distinct edges the agent can traverse during the search: when this bound is violated, the agent stops the search and goes back to the node it occupied at the beginning of the search. The return value of {\tt CDFS}$(l,b)$ is a couple $(n,T)$. The first term $n$ is an integer such that $b-n$ is the number of distinct edges that have been traversed during the execution of {\tt CDFS}$(l,b)$. If the bound $b$ has been respected then $n\geq 0$, otherwise $n=-1$. Concerning the second term $T$ of the return value, it simply corresponds to the resulting DFS tree of the execution of {\tt CDFS}$(l,b)$. If $n\geq0$ and $v$ is the occupied node at the start of {\tt CDFS}$(l,b)$, then for every node $u$ such that $d_{T}(u,v)<l$, $u$ is complete in $\mathcal{M}$ at the end of {\tt CDFS}$(l,b)$. Note that in the particular case where $l=1$ and $m=\perp$ in Algorithm~\ref{alg:GE}, the second argument of each call to {\tt CDFS} is always set to the degree of the node from which the function is executed (cf. line~\ref{secondtype1} of Algorithm~\ref{alg:GE}) in order to ensure that this node becomes complete in $\mathcal{M}$ at the end of the call.

\begin{algorithm}[H]
\caption{${\tt CDFS}(l, b)$\label{alg:CDFS}}
\SetKwFunction{MT}{MoveTo}
\SetKwFunction{CDFS}{CDFS}
\SetKwFunction{MG}{Merge}
$v :=$ the current node; $T := (\{v\},\emptyset)$; $bound := b$\;
\If{$l>0$}
{
Mark node $v$\;
\While{node $v$ is incomplete in $\mathcal{M}$ \AND $bound\geq 0$}
{
$pt_1 :=$ the smallest free port at node $v$ in $\mathcal{M}$\;\label{CDFSadd1}
Take port $pt_1$\;
$w :=$ the current node\;
                        $pt_2 :=$ the port by which the agent has just entered node $w$\;
\eIf{$v<w$}{
$K := (\{v,w\},\{(v,w,pt_1,pt_2)\})$\;
}
{
$K := (\{v,w\},\{(w,v,pt_2,pt_1)\})$\;
}
                        $\mathcal{M} := \mathcal{M}\sqcup K$; $bound := bound-1$\;\label{CDFSadd2}
\If{$w$ is not marked}{
$(bound,T') :=$ \CDFS{$l-1,bound$}\;
$T := T \sqcup T' \sqcup K$\;
}

       Take port $pt_2$\;

}
Unmark node $v$\;
}
\Return $(bound,T)$\;
\end{algorithm}

The case where $l\geq2$ and $m\ne\perp$ in Algorithm~\ref{alg:GE} relies on function {\tt LocalExpansion}. It is exactly here that we make use of the algorithmic technique of \cite{DKK} mentioned at the end of Section~\ref{sec:int}, which is based on a set of adequately pruned trees. In our solution, this set corresponds to the variable $\mathcal{T}$. It is a global variable like $\mathcal{M}$ and it is initialized to $\emptyset$ at the beginning of each call to {\tt GlobalExpansion} (cf. line~\ref{GE:initialize} of Algorithm~\ref{alg:GE}).
Let us consider the $i$th call $LE_i$ to {\tt LocalExpansion}$(l,b)$ made from node $L[i]$ during an execution of {\tt GlobalExpansion}$(l,m)$. At the end of $LE_i$, the return value of {\tt LocalExpansion}$(l,b)$ is an integer $n\geq-1$ such that $b-n$ is the number of distinct edges that have been traversed during $LE_i$ and that were not in $\mathcal{M}$ at the start of $LE_i$. Besides, in the case where $n\geq0$, at the end of $LE_i$ we can guarantee that for each incomplete node $u$ of $\mathcal{M}$, $d_{\mathcal{M}}(L[i],u)> l$ or $u$ is one of the last $|L|-i$ nodes of $L$ (i.e., a node of $L$ from which the agent has not yet executed {\tt LocalExpansion}$(l,b)$).

To see the algorithmic technique in question at work, let us focus on an iteration $I$ of the first while loop of Algorithm~\ref{alg:LE} occuring in $LE_i$. This iteration starts at node $L[i]$ and we will show in Section~\ref{sec:proof} that at the beginning of $I$, we necessarily have the following properties.
\vspace{-0.3cm}
\begin{itemize}
\item $\mathcal{T}$ is a set of node disjoint trees that are all subgraphs of $\mathcal{M}$.
\item For each tree $Tr$ of $\mathcal{T}$, $|Tr| \geq \lfloor\frac{l}{8}\rfloor$ if $Tr$ contains a node different from $L[i]$.
\item Every incomplete node of $\mathcal{M}$ belongs to a tree of $\mathcal{T}$ or is one of the last $|L|-i$ nodes of $L$.
\end{itemize}
\vspace{-0.2cm}
\begin{algorithm}[H]
\caption{${\tt LocalExpansion}(l, b)$\label{alg:LE}}
\SetKwFunction{CIN}{IncompleteNodes}
\SetKwFunction{N}{Nodes}
\SetKwFunction{PRUNE}{Prune}
\SetKwFunction{EXPLORE}{Explore}
\SetKwFunction{MT}{MoveTo}
\SetKwFunction{MG}{Merge}
$bound := b$; $v :=$ the current node\;\label{LE:beforeif}
\If{$v$ is incomplete in $\mathcal{M}$ \AND no tree of $\mathcal{T}$ contains node $v$}{\label{LE:if}
$\mathcal{T} := \mathcal{T}\cup \{(\{v\},\emptyset)\}$\;\label{LE:then}
}
\While{${\tt IncompleteNodes}(v,\mathcal{M},l) \cap {\tt Nodes}(\mathcal{T})\ne\emptyset$ \AND $bound\geq0$}{\label{LE:condition}
        $u :=$ the node with the smallest label in ${\tt IncompleteNodes}(v,\mathcal{M},l) \cap {\tt Nodes}(\mathcal{T})$\;\label{LE:conditionu}
        \MT{$\mathcal{M},u$}\;\label{LE:move}
        \PRUNE{$l$}\;\label{LE:prune}
        $bound :=$ \EXPLORE{$l, bound$}\;\label{LE:explore}
Remove from $\mathcal{T}$ every tree for which all the nodes are complete in $\mathcal{M}$\;\label{clean1}
        \While{there are two trees $T$ and $T'$ in $\mathcal{T}$ having a common node}{\label{clean20}
        $T'' :=$ the spanning tree produced by the BFS traversal of $T\sqcup T'$ from the node having the smallest label in $T\sqcup T'$\;
        $\mathcal{T} := (\mathcal{T}\setminus\{T,T'\})\cup\{T''\}$\;\label{clean2}
        }
Execute in the reverse order all the edge traversals that have been made since the beginning of the current iteration of the while loop\;\label{LE:backtrack}
}
\Return $bound$\;
\end{algorithm}

Let us examine what happens during iteration $I$. At the beginning of $I$, the agent follows a path of length at most $l$ from node $L[i]$ to a node $u$ that is incomplete in $\mathcal{M}$ (cf. line~\ref{LE:conditionu} of Algorithm~\ref{alg:LE}). By the first and third properties and the condition at line~\ref{LE:condition} of Algorithm~\ref{alg:LE}, node $u$ belongs to a unique tree $T_u\subseteq\mathcal{G}$ of $\mathcal{T}$. Once the agent occupies node $u$, the tree $T_u$ is pruned via the procedure {\tt Prune}$(l)$ at line~\ref{LE:prune} of Algorithm~\ref{alg:LE}. The pseudocode of procedure {\tt Prune} is detailed in Algorithm~\ref{alg:PR}. 

\begin{algorithm}[H]
\caption{${\tt Prune}(l)$\label{alg:PR}}
$v :=$ the current node\;
$T_v :=$ the tree of $\mathcal{T}$ containing node $v$\;
$\mathcal{T} := \mathcal{T}\setminus\{T_v\}$\;
Root $T_v$ at node $v$\;
\ForEach{node $u$ of $T_v$ such that $d_{T_v}(u,v)=\max\{1,\lfloor\frac{l}{4}\rfloor\}$}{
$T_u :=$ the subtree of $T_v$ rooted at $u$\;
\If{$\epsilon_{T_u}(u)\geq \lfloor\frac{l}{4}\rfloor-1$}
{
$\mathcal{T} := \mathcal{T}\cup\{T_u\}$\;
Remove from $T_v$ all nodes that belong to $T_u$ and all edges that are incident to a node of $T_u$\;
}
}
$\mathcal{T} := \mathcal{T}\cup \{T_v\}$\;
\end{algorithm}

In the context of iteration $I$, the pruning operation will transform $T_u$ into a tree $T'_u$ such that $\epsilon_{T'_u}(u)\leq \lfloor\frac{l}{2}\rfloor-1$, while preserving the three properties listed above: this offers two important advantages to which we will return at the end of this section. Once the pruning is done, the agent applies function {\tt Explore}$(l,bound)$, whose pseudocode is given in Algorithm~\ref{alg:EX}.

\begin{algorithm}[H]
\caption{${\tt Explore}(l, b)$\label{alg:EX}}
\SetKwFunction{MT}{MoveTo}
\SetKwFunction{CDFS}{CDFS}
$bound := b$; $i := 1$; $v :=$ the current node\;
$T :=$ the tree of $\mathcal{T}$ containing node $v$\;
$V :=$ array containing all the nodes of $T$ sorted in the order of the first visit through the DFS traversal of $T$ from node $v$\;
\While{$i\leq |V|$ \AND $bound\geq0$}{\label{EX:while}
\MT{$T,V[i]$}\;\label{EX:move}
\If{node $V[i]$ is incomplete in $\mathcal{M}$}{\label{EX:callCDFS0}
$(bound, T') := $ \CDFS{$\lfloor \frac{l}{2}\rfloor, bound$}\;\label{EX:callCDFS}
$\mathcal{T} := \mathcal{T}\cup \{T'\}$\;\label{EX:add}
}
}
\Return $bound$\;
\end{algorithm}

In the pseudocodes of {\tt LocalExpansion} and of {\tt Explore}, variable $bound$ corresponds at any stage to the number of remaining edges the agent is authorized to traverse outside of $B_{f}(\mathcal{G},s)$. In the context of iteration $I$, function {\tt Explore}$(l,bound)$ permits the agent to explore tree $T'_u$ and to execute function {\tt CDFS}$(\lfloor\frac{l}{2}\rfloor,bound)$ from the nodes of $T'_u$ that are incomplete in $\mathcal{M}$, as long as variable $bound$ remains non-negative. These executions of {\tt CDFS} occuring during the exploration of $T'_u$ create in turn trees that are added to $\mathcal{T}$ (cf. line~\ref{EX:add} of Algorithm~\ref{alg:EX}) and that contain the new incomplete nodes of $\mathcal{M}$. If the return value of function {\tt Explore}$(l,bound)$ is non-negative, we will show in Section~\ref{sec:proof} that all the nodes of $T'_u$ have become complete in $\mathcal{M}$. Under the same condition, we will also guarantee that each tree $Tr$, which has been added to $\mathcal{T}$ during the execution of function {\tt Explore}, contains an incomplete node only if $|Tr|\geq \lfloor\frac{l}{8}\rfloor$. Both these guarantees combined with lines~\ref{clean1} to~\ref{clean2} of Algorithm~\ref{alg:LE} will allow us to show that our three properties will be satisfied for the next iteration $I'$, if any, even if it occurs in another call to {\tt LocalExpansion} (in the same execution of {\tt GlobalExpansion}$(l,m)$). In particular, this is made possible by the fact that $\mathcal{T}$ is never reset between the calls to {\tt LocalExpansion} during the execution of the while loop of Algorithm~\ref{alg:GE}.

To fully appreciate the process accomplished during $I$, we need to come back to the two aforementioned advantages that are brought by the pruning operation. The first advantage concerns the height of $T'_u$. The fact that $\epsilon_{T'_u}(u)\leq \lfloor\frac{l}{2}\rfloor-1$ is a key element to control the maximal distance between the agent and node $s$. Without this, the agent could go too far from node $s$ and we would not be able to guarantee that the agent explores only edges of $B_{f+2l-1}(\mathcal{G},s)$ during the execution of {\tt GlobalExpansion}$(l,m)$ (which is a crucial property as pointed out in Section~\ref{sec:int}). The second advantage concerns the size of $T'_u$. The pruning operation preserves the second property, and thus (1) $T'_u$ corresponds to a tree containing only node $L[i]$ or (2) $|T'_u|\geq\lfloor\frac{l}{8}\rfloor$. This implies that the cost resulting from the moves of line~\ref{LE:move} of Algorithm~\ref{alg:LE} and line~\ref{EX:move} of Algorithm~\ref{alg:EX} is linear in the size of $T'_u$. Besides, if $bound$ is still non-negative at the end of {\tt Explore}$(l,bound)$, all the nodes of $T'_u$ have become complete (it is in particular the case for node $u$) and the tree is removed from $\mathcal{T}$ through line~\ref{clean1} of Algorithm~\ref{alg:LE}. After this removal, no edge of $T'_u$ will be an edge of another tree of $\mathcal{T}$ till the end of the execution of {\tt GlobalExpansion}$(l,m)$. As a result, if the return value of {\tt Explore}$(l,bound)$ is non-negative in $I$, we can associate the moves of line~\ref{LE:move} of Algorithm~\ref{alg:LE} and line~\ref{EX:move} of Algorithm~\ref{alg:EX} to at least one node that becomes complete during $I$ and to at least $\lfloor \frac{l}{8}\rfloor$ edges that will no longer be edges of any tree of $\mathcal{T}$ till the end of the execution of {\tt GlobalExpansion}$(l,m)$. In our analysis, this association will enable us to amortize efficiently the number of times the agent retraverses the edges that have been already explored during any previous iteration of the considered while loop. This will be a decisive argument to show the cost of $\mathcal{O}(e(f)+m)$ for the execution of {\tt GlobalExpansion}$(l,m)$ in the case where $l\geq2$ and $m\ne\perp$.

\vspace{-0.3cm}
\section{Correctness and complexity analysis}
\label{sec:proof}
\vspace{-0.2cm}
This section is dedicated to the proof of correctness and of complexity of Algorithm {\tt TreasureHunt}$(x)$ in the unrestricted model.  {\tt TreasureHunt}$(x)$ is an exploration algorithm that can be executed also if there is no treasure in $\mathcal{G}$. We first establish several exploration properties of our algorithm or of its components assuming that there is no treasure in $\mathcal{G}$. In fact, this assumption concerns all the lemmas (and only them) of this section and it will not be repeated in their statements in order to lighten the presentation. After the series of lemmas, we show the main result of this section, namely Theorem~\ref{theo:final0}, which specifies that our algorithm allows to find the treasure at a cost quasi-linear in $e(d)$. 

Throughout the proof of correctness, we will often have to consider the value of the global variable $\mathcal{M}$ before or after some executions. To this end, we introduce the following convention: given an execution $\mathcal{E}$ of Algorithm ${\tt TreasureHunt}(x)$ or some part of it, we denote by $M_1(\mathcal{E})$ the value of $\mathcal{M}$ at the beginning of $\mathcal{E}$ and by $M_2(\mathcal{E})$ the value of $\mathcal{M}$ at the end of $\mathcal{E}$.

We start by giving two lemmas concerning the function ${\tt CDFS}(l, b)$. They list some properties that will be useful in the sequel. They are direct consequences of Algorithm~\ref{alg:CDFS} and can be easily proved by induction on $l$.
\vspace{-0.1cm}
\begin{lemma}
\label{lem:CDFX}
Consider an execution $\mathcal{E}$ of function ${\tt CDFS}(l, b)$ from a node $u$ of $\mathcal{G}$ where $l\geq1$ and $b\geq0$ are integers. Assume that $M_1(\mathcal{E})\subseteq \mathcal{G}$. Execution $\mathcal{E}$ terminates at node $u$, and the agent always knows a path of length at most $l$ from node $u$ to its current node during $\mathcal{E}$.
\end{lemma}
\vspace{-0.2cm}
\begin{lemma}
\label{lem:CDFX2}
Consider an execution $\mathcal{E}$ of function ${\tt CDFS}(l, b)$ from a node $u$ of $\mathcal{G}$ where $l\geq1$ and $b\geq0$ are integers. Assume that $M_1(\mathcal{E})\subseteq \mathcal{G}$. Function ${\tt CDFS}(l, b)$ returns a couple $(i,Tr)$ such that the following properties are satisfied.
\vspace{-0.1cm}
\begin{itemize}
\item Let $G$ be the subgraph of $\mathcal{G}$ that has been explored during $\mathcal{E}$. $G\subseteq B_l(\mathcal{G},u)$, $|M_1(\mathcal{E})\sqcap G|=0$, $M_1(\mathcal{E})\sqcup G = M_2(\mathcal{E})$, $Tr$ is a spanning tree of $G$ and $i=b-|G|\geq -1$.
\item The cost of $\mathcal{E}$ is $2|G|$ and $\epsilon_{Tr}(u)\leq l$.
\item If $i\geq0$ then for every node $v$ of $Tr$ such that $d_{Tr}(u,v)<l$, $v$ is complete in $M_2(\mathcal{E})$. If $i=-1$, then there exists a node $v$ of $Tr$ such that $d_{Tr}(u,v)\leq l-1$ and $v$ is incomplete in $M_2(\mathcal{E})$.
\end{itemize}
\end{lemma}

The following lemma establishes the properties of function ${\tt GlobalExpansion}(l,m)$ that will be used to prove Lemma~\ref{lem:search1} that concerns procedure {\tt Search}$(x)$.

\begin{lemma}
\label{lem:GE2}
Consider an execution $\mathcal{E}$ of function ${\tt GlobalExpansion}(l,m)$ from the source node $s$, where $l$ is a positive integer and $m$ is either a positive integer or $\perp$. Assume that $M_1(\mathcal{E})=B_f(\mathcal{G},s)$ for some integer $f\geq0$.
\begin{itemize}
\item if $m\ne\perp$, or $m=\perp$ and $l=1$, then $\mathcal{E}$ terminates at node $s$ and during $\mathcal{E}$ the agent always knows a path in $\mathcal{G}$ of length at most $f+2l-1$ from node $s$ to its current node.
\item If $m=\perp$ and $l=1$ then the cost of $\mathcal{E}$ is $\mathcal{O}(e(f+1))$ and $B_{f+1}(\mathcal{G},s)= M_2(\mathcal{E})$.
\item If $m\ne\perp$ and function ${\tt GlobalExpansion}(l,m)$ returns true (resp. false) then $B_{f+l}(\mathcal{G},s)\subseteq M_2(\mathcal{E})$ (resp. $B_{f}(\mathcal{G},s)\subseteq M_2(\mathcal{E})$ and $e(f+2l-1)>e(f)+m$) and the cost of $\mathcal{E}$ is $\mathcal{O}(e(f)+m)$.
\end{itemize}
\end{lemma}
\begin{proof}
First observe that the global variable $\mathcal{M}$ is always a subgraph of $\mathcal{G}$ during $\mathcal{E}$. This comes from the fact that $M_1(\mathcal{E})=B_f(\mathcal{G},s)$, and from lines~\ref{CDFSadd1} to~\ref{CDFSadd2} of Algorithm~\ref{alg:CDFS} that are the only places where $\mathcal{M}$ may be modified during $\mathcal{E}$. By line~\ref{DFS} of Algorithm~\ref{alg:GE}, $L$ is an array containing all the nodes of the boundary of $B_f(\mathcal{G},s)$. The nodes of $B_f(\mathcal{G},s)$ that are not in $L$ are necessarily complete in $M_1(\mathcal{E})$ (or otherwise we get a contradiction with the fact that $M_1(\mathcal{E})=B_f(\mathcal{G},s)$).

According to Algorithm~\ref{alg:GE}, the edge traversals made during $\mathcal{E}$ can be divided into three distinct types. The first type corresponds to those that aim to position the agent at each node of $L$ (cf. line~\ref{firsttype} of Algorithm~\ref{alg:GE}). The second type consists of edge traversals that aim to expand $\mathcal{M}$ (cf. lines~\ref{secondtype1}, \ref{secondtype1bis} and~\ref{secondtype2} of Algorithm~\ref{alg:GE}).  The third type consists of edge traversals that permit to relocate the agent at node $s$ at the end of $\mathcal{E}$ (cf. line~\ref{thirdtype} of Algorithm~\ref{alg:GE}). The total number of edge traversals made by the agent will be analysed below and will vary according to different cases. However, with a few arguments, we can already give some properties of the first and third types, in particular concerning their order of magnitude.

Note that an execution of {\tt CDFS}$(1,deg(L[i]))$ (resp. {\tt CDFS}$(1,b)$) in Algorithm~\ref{alg:GE} from a node $L[i]$ starts and ends at $L[i]$ according to Lemma~\ref{lem:CDFX}. Also note that if an execution of {\tt LocalExpansion} from $L[i]$ terminates (this will be shown below), the agent is back at $L[i]$ at the end of this execution in view of line~\ref{LE:backtrack} of Algorithm~\ref{alg:LE}. From the above explanations, it follows that for all $2\leq i\leq|L|$, the agent is at node $L[i-1]$ when it starts executing the $i$th iteration of the while loop of Algorithm~\ref{alg:GE}. It also follows that the agent is at node $L[|L|]$ when it starts executing line~\ref{thirdtype} of Algorithm~\ref{alg:GE}. As a result, we have the following claim owing to the fact that the nodes of $L$ are sorted in the order of the first visit through the DFS traversal of $B_f(\mathcal{G},s)$ from node $s$, and the fact that the agent always takes the shortest path in the tree produced by this traversal when it executes the move instructions of lines~\ref{firsttype} and~\ref{thirdtype} of Algorithm~\ref{alg:GE}.

\begin{claim}
\label{claim:GE1}
The total cost induced by the edge traversals belonging to the first or third type in $\mathcal{E}$ is $\mathcal{O}(e(f))$. Moreover, at the beginning and at the end of each edge traversal of the first or third type in $\mathcal{E}$, the agent knows a path of length at most $f$ from node $s$ to its current node. Finally, if $l=1$ or if each execution of {\tt LocalExpansion} in $\mathcal{E}$ terminates, then $\mathcal{E}$ terminates at node $s$.
\end{claim}

We first prove the lemma in the easiest case where $l=1$. The nodes of $L$ are all at distance at most $f$ from node $s$ in $B_f(\mathcal{G},s)$. If $m=\perp$, the agent executes function ${\tt CDFS}(1,deg(L[i]))$ from each node $L[i]$. If $m\ne\perp$, the agent executes ${\tt CDFS}(1,b)$
also from the nodes of $L$, but not necessarily all of them if $b$ becomes negative. Let $G$ be the subgraph of $\mathcal{G}$ that is explored during the executions of function {\tt CDFS}, whether $m$ is $\perp$ or not. Using Lemma~\ref{lem:CDFX2}, it follows by induction on the number of calls to {\tt CDFS} that $M_2(\mathcal{E})=M_1(\mathcal{E})\sqcup G \subseteq B_{f+1}(\mathcal{G},s)$, $|B_f(\mathcal{G,s})\sqcap G|=0$ and the total cost induced by these calls is $2|G|$: in particular, if $m\ne\perp$, variable $b$ is equal to $m-|G|\geq-1$ at the end of $\mathcal{E}$. Using the same lemma, it follows that all the nodes of $L$ are complete in $M_2(\mathcal{E})$ if $m=\perp$, or if $m\ne\perp$ and $b\geq0$ at the end of $\mathcal{E}$. Finally, at the beginning and at the end of each move made during each execution of {\tt CDFS} from a node $L[i]$, the agent knows a path of length at most $1$ from $L[i]$.

From Claim~\ref{claim:GE1} and the above explanations, it follows that $\mathcal{E}$ terminates at node $s$ and during $\mathcal{E}$ the agent always knows a path in $\mathcal{G}$ of length at most $f+1$ from node $s$ to its current node. It also follows that if $m=\perp$ (resp. $m\ne\perp$) the cost of $\mathcal{E}$ is $\mathcal{O}(e(f+1))$ (resp. $\mathcal{O}(e(f)+m)$ as $|G|\leq m+1$). Besides, if $m=\perp$, or $m\ne\perp$ and $b\geq0$ at the end of $\mathcal{E}$, then $B_{f+1}(\mathcal{G},s)= M_2(\mathcal{E})$ and the return value of ${\tt GlobalExpansion}(l,m)$ is true according to line~\ref{GE:lastline} of Algorithm~\ref{alg:GE}. Otherwise, $B_{f}(\mathcal{G},s)\subseteq M_2(\mathcal{E})$, the return value of ${\tt GlobalExpansion}(l,m)$ is false, and $e(f+1)>e(f)+m$ (since the last value of $b$ is $-1$ and $|B_f(\mathcal{G},s)\sqcap G|=0$, we have $|M_1(\mathcal{E})\sqcup G|=e(f)+m+1$). This proves the lemma in the case where $l=1$.

Let us now turn our attention to the main case where $m\ne \perp$ and $l\geq2$. The analysis of this case will involve the global variable $\mathcal{T}$. Strictly speaking, the value of this variable will be always a set of trees. However, if a node (resp. edge) belongs to a tree of $\mathcal{T}$, we will sometimes say by abuse of language that it is a node (resp. an edge) of $\mathcal{T}$. Still by abuse of language, we will sometimes say that a node (resp. an edge) has been removed from $\mathcal{T}$, if at some point this node (or this edge) no longer belongs to any tree of $\mathcal{T}$. The first while loop of Algorithm~\ref{alg:LE} will be called $W1$, and when we speak of variable $bound$, it will be always implied it is the variable of Algorithm~\ref{alg:LE} unless explicitly mentioned otherwise. The $i$th execution of function {\tt LocalExpansion}$(l,b)$, if any, made within $\mathcal{E}$ will be denoted by $LE_i$. The starting node of $LE_i$ is node $L[i]$ according to line~\ref{firsttype} of Algorithm~\ref{alg:GE} and line~\ref{LE:backtrack} of Algorithm~\ref{alg:LE}.

We start with two claims. 

\begin{claim}
\label{claim:GE2}
For every $1\leq i\leq |L|$, execution $LE_i$ terminates, and during it the agent always knows a path of length at most $2l-1$ from node $L[i]$ to its current node. Moreover, at the beginning and at the end of each iteration of $W1$ made during $LE_i$, we have the following three properties.
\begin{itemize}
\item The agent is at node $L[i]$.
\item Variable $\mathcal{T}$ is a set of node disjoint trees that are all subgraphs of $\mathcal{M}$. 
\item Every incomplete node in $\mathcal{M}$ is a node of $\mathcal{T}$ or one of the last $|L|-i$ nodes of $L$.
\end{itemize}
\end{claim}

\begin{proofclaim}
Consider an integer $1\leq i\leq |L|$ and suppose that the agent ends up executing $LE_i$. Let us first analyse what happens during a given iteration $I$ of $W1$ made during $LE_i$, assuming that at the beginning of the considered iteration the three properties of the claim are satisfied. The existence of $I$ implies that $bound$ is non-negative at the start of $I$.

Once the execution of the move instruction of line~\ref{LE:move} of Algorithm~\ref{alg:LE} has been made, the agent occupies a node $u$ of $M_1(I)$. In view of line~\ref{LE:conditionu}, node $u$ is incomplete in $M_1(I)$ and the agent has reached it by following a path of length at most $l$ from $L[i]$. By assumption, node $u$ belongs to a unique tree $T_u$ of $\mathcal{T}$ and $T_u\subseteq M_1(I)$. By Algorithm~\ref{alg:PR}, we know that after the execution of function {\tt Prune}$(l)$ at line~\ref{LE:prune} of Algorithm~\ref{alg:LE}, $\mathcal{T}$ is unchanged, or $T_u$ has been replaced in $\mathcal{T}$ by smaller node disjoint trees that are all subgraphs of $T_u$ and whose union spans $T_u$. In particular, we have $\epsilon_{T'_u}(u)\leq \lfloor \frac{l}{2}\rfloor-1$, where $T'_u$ is the tree of $\mathcal{T}$ containing node $u$ after the pruning operation. After that, the agent executes function {\tt Explore}$(l,bound)$ from node $u$. By Algorithm~\ref{alg:EX}, this execution  consists of a traversal of $T'_u$ interlaced with executions of {\tt CDFS} from the incomplete nodes of $T'_u$ (from each of these nodes, the distance in $M_1(I)$ to node $L[i]$ and to node $s$ are respectively at most $l+\lfloor \frac{l}{2}\rfloor-1$ and at most $f+l+\lfloor \frac{l}{2}\rfloor-1$). The first parameter of these executions of {\tt CDFS} is $\lfloor \frac{l}{2}\rfloor$. Hence, in view of Lemma~\ref{lem:CDFX2}, we know that after the execution of function {\tt Explore} in line~\ref{LE:explore} of Algorithm~\ref{alg:LE}, variable $bound\geq-1$ and variable $\mathcal{M}$ has been extended by the subgraph $K$ of $\mathcal{G}$ that has been explored in the calls to {\tt CDFS} made during {\tt Explore}$(l,bound)$. In view of the same arguments, we know that the trees added to $\mathcal{T}$ during the execution of function {\tt Explore} are all subgraphs of $\mathcal{M}$ at the end of this execution: precisely, the union of these added trees forms a spanning subgraph
 of $K$, and thus the third property of the claim is still satisfied. In addition, during the execution of {\tt Explore}$(l,bound)$, the agent always knows a path of length at most $l-1$ from node $u$ to its current node, due to Lemma~\ref{lem:CDFX} and to the fact that $\epsilon_{T'_u}(u)\leq \lfloor \frac{l}{2}\rfloor-1$. Note that once the process of line~\ref{LE:explore} of Algorithm~\ref{alg:LE} is over, $\mathcal{T}$ is indeed still a set of trees that are all subgraphs of $\mathcal{M}$, but some trees of $\mathcal{T}$ may be not node disjoint. This is resolved through the executions of lines~\ref{clean1} to~\ref{clean2} of Algorithm~\ref{alg:LE}, that will permit to guarantee the second property of the claim while preserving the third property.

Finally, taking into account line~\ref{LE:backtrack} of Algorithm~\ref{alg:GE} that consists of an execution in the reverse order of all the edge traversals that have been previously made in $I$, it follows from the above explanations that $I$ terminates and the agent always knows a path of length at most $2l-1$ from $L[i]$ to its current node during $I$. It also follows that the conditions that are supposed to be satisfied at the beginning of $I$, are still satisfied at the end of $I$. If variable $bound$ is negative at the end of $I$, there will be no more iterations of $W1$ thereafter in $LE_i$ and even in $\mathcal{E}$.

Note that, outside of $W1$, the position of the agent does not change within $LE_i$. It is also the case for variable $\mathcal{T}$, except just before $W1$ where the node $L[i]$ is added into $\mathcal{T}$ if and only if $L[i]$ is incomplete in $\mathcal{M}$ and $L[i]$ belongs to no tree of $\mathcal{T}$. Also note that at the beginning of the first call to function {\tt LocalExpansion}$(l,b)$, we have $b=m$, $\mathcal{T}=\emptyset$ and the agent occupies node $L[1]$. In particular, this implies that the initial assumptions made for the analysis of iteration $I$ are satisfied just before the execution of $W1$ in the first call to function {\tt LocalExpansion}. Hence, it follows by induction on $i$, that all the statements of the claim hold, except the statement that $LE_i$ terminates if $bound$ is never negative. More specifically, it can be shown that each iteration of $W1$ in $LE_i$ terminates, but at this point of the proof, we are not yet sure that the number of these iterations is finite if $bound$ is never negative.


So, to conclude the proof of this claim, it remains to show that the number of iterations of $W1$ in $LE_i$ is finite assuming that variable $bound$ is never negative in $LE_i$. During an iteration $I$ of $LE_i$, the node $u$ to which the agent moves when executing line~\ref{LE:move} of Algorithm~\ref{alg:LE} is necessarily incomplete in $M_1(I)$ and such that $d_{M_1(I)}(L[i],u)\leq l$. {Using Lemma~\ref{lem:CDFX2}, we can prove that node $u$ becomes complete in $\mathcal{M}$ after the first call to {\tt CDFS} made within the execution of ${\tt Explore}(l,bound)$ in $I$ because the return value of this execution, which becomes by then the value of $bound$, is necessarily at least $0$ (otherwise we get a contradiction with the assumption that variable $bound$ is never negative in $LE_i$)}. Since the number of nodes $u$ such that $d_\mathcal{G}(l[i],u)\leq l$ is finite, it follows that the number of iterations in $LE_i$ is finite even if variable $bound$ is never negative in $LE_i$. This concludes the proof of the claim.
\end{proofclaim}

In view of Claims~\ref{claim:GE1} and~\ref{claim:GE2}, we are guaranteed that $\mathcal{E}$ terminates (at node $s$) and the number of calls to {\tt CDFS} made during $\mathcal{E}$ is finite. As pinpointed in the proof of Claim~\ref{claim:GE2}, these calls are triggered only through the executions of function {\tt Explore} from nodes that necessarily belong to $B_{f+l+\lfloor \frac{l}{2} \rfloor-1}(\mathcal{G},s)$. Hence, using Lemma~\ref{lem:CDFX2}, the next claim can be shown by induction on the number of calls to {\tt CDFS} made during $\mathcal{E}$. In this claim and in the rest of this proof, $H$ denotes the subgraph of $\mathcal{G}$ consisting of all the edges and nodes that have been visited by the agent during the calls to {\tt CDFS} within $\mathcal{E}$.

\begin{claim}
\label{claim:GE4}
The total cost of the executions of {\tt CDFS} in $\mathcal{E}$ is $2|H|$, $M_2(\mathcal{E})=M_1(\mathcal{E})\sqcup H\subseteq B_{f+2l-1}(\mathcal{G},s)$ and $|M_1(\mathcal{E})\sqcap H|=0$. Moreover, at the end of $\mathcal{E}$, the value of variable $b$ in Algorithm~\ref{alg:GE} is $m-|H|\geq -1$.
\end{claim}

Note that in view of Claims~\ref{claim:GE1} and~\ref{claim:GE2}, the agent always knows a path of length at most $f+2l-1$ from node $s$ to its current node during $\mathcal{E}$. Suppose that $m-|H|=-1$. In this case, function ${\tt GlobalExpansion}(l,m)$ returns false, and $B_{f}(\mathcal{G},s)\subseteq M_2(\mathcal{E})$ as $M_2(\mathcal{E})=B_f(\mathcal{G},s)\sqcup H$ by Claim~\ref{claim:GE4}. Moreover, $M_1(\mathcal{E})\sqcup H\subseteq B_{f+2l-1}(\mathcal{G},s)$ and $|B_f(\mathcal{G},s)\sqcap H|=0$ by Claim~\ref{claim:GE4}, and $|H|>m$, which implies that $e(f+2l-1)> e(f)+m$.

Now, suppose that that $m-|H|\geq 0$. In this case, variable $bound$ is never negative and function ${\tt GlobalExpansion}(l,m)$ returns true. Moreover, we have the following claim.

\begin{claim}
$B_{f+l}(\mathcal{G},s)\subseteq M_2(\mathcal{E})$.
\end{claim}

\begin{proofclaim}
Assume by contradiction that the claim does not hold. Since $M_2(\mathcal{E})=B_f(\mathcal{G},s) \sqcup H$, it follows that there exist an integer $k$ and a node $u$ of $M_2(\mathcal{E})$ such that $d_{M_2(\mathcal{E})}(L[k],u)\leq l$ and $u$ is incomplete in $M_2(\mathcal{E})$. Without loss of generality, suppose that in $M_2(\mathcal{E})$, $L[k]$ is the node of $L$ that is the closest (or one the closest) from node $u$. This implies that, at the end of each iteration of $W1$ made during $LE_k$, there is a node $v$ that does not belong to the last $|L|-k$ nodes of $L$, that is incomplete in $\mathcal{M}$ and that is such that $d_{\mathcal{M}}(L[k],v)\leq l$. Hence, the execution $LE_k$ never terminates, as the condition of $W1$ (cf. line~\ref{LE:condition} of Algorithm~\ref{alg:LE}) always evaluates to true during $LE_k$ in view of Claim~\ref{claim:GE2} and of the fact that variable $bound$ is never negative. We then get a contradiction with the fact that $\mathcal{E}$ terminates, which concludes the proof of the claim.
\end{proofclaim}

Consequently, to end the analysis of the current case (and thus the proof of this lemma), it remains to prove that the cost that has been paid during $\mathcal{E}$ is $\mathcal{O}(e(f)+m)$. More precisely, in view of Claim~\ref{claim:GE1}, it is enough to show that the number of moves of the second type (which correspond here to the moves made during the executions of function ${\tt LocalExpansion}$) belongs to $\mathcal{O}(m)$. Actually, a first step has been made via Claim~\ref{claim:GE4} that implies that the number of moves of the second type made during the executions of {\tt CDFS} is at most $2(m+1)$. Hence, we just have to prove that the number of moves of the second type made outside of the executions of {\tt CDFS} and outside of the executions of line~\ref{LE:backtrack} of Algorithm~\ref{alg:LE} is $\mathcal{O}(m)$. These remaining moves are of two kinds. The first kind concerns those that are made during the execution of line~\ref{EX:move} of Algorithm~\ref{alg:EX} in order to make the DFS traversal of a tree of $\mathcal{T}$: they will be called the \emph{blue moves}. The second kind concerns those that are made during the execution of line~\ref{LE:move} of Algorithm~\ref{alg:LE}: they will be called the \emph{red moves}. To conduct the discussions, we need two more claims.

\begin{claim}
\label{claim:remove}
Suppose that at some time $t$ during $\mathcal{E}$, a tree is removed from $\mathcal{T}$ via the execution of line~\ref{clean1} of Algorithm~\ref{alg:LE}. For every edge $e$ of the removed tree, $e$ will not be in $\mathcal{T}$ from time $t$ to the end of $\mathcal{E}$.
\end{claim}

\begin{proofclaim}
The edges of $\mathcal{T}$ always originally come from the trees returned by the calls to function {\tt CDFS}. According to Lemma~\ref{lem:CDFX2}, for any of these calls $F$, the returned tree $Tr$ has no common edge with $M_1(F)$ and $M_1(F)\sqcup Tr\subseteq M_2(F)$. This implies that at any point of execution $\mathcal{E}$, every edge $e$ of $\mathcal{G}$ belongs to at most one tree of $\mathcal{T}$ (the operation of pruning and merging of lines~\ref{LE:prune} and~\ref{clean2} of Algorithm~\ref{alg:LE} can never change this fact). This also implies that for every edge $e$ of $\mathcal{G}$, there is at most one call to function {\tt CDFS} in $\mathcal{E}$ that returns a tree containing edge $e$. Hence, when a tree is removed from $\mathcal{T}$, none of its edges will ever appear again in $\mathcal{T}$ from the time of the removal to the end of $\mathcal{E}$. This proves the claim.
\end{proofclaim}


\begin{claim}
\label{claim:size}
At the beginning of each execution of function {\tt Explore}$(l,bound)$ in $LE_i$, we have the following property $P(i)$: for each tree $Tr$ of $\mathcal{T}$, $|Tr| \geq \lfloor\frac{l}{8}\rfloor$ if $Tr$ contains a node different from $L[i]$.
\end{claim}

\begin{proofclaim}
Suppose by contradiction that the claim does not hold and suppose that $i$ is the smallest integer for which the claim is not verified. Let $I$ be the first iteration in $LE_i$ such that $P(i)$ is not satisfied at the beginning of the execution of function {\tt Explore}$(l,bound)$ in $I$. Let $u$ (resp. $T_u$) be the node (resp. the tree of $\mathcal{T}$) in which the agent is located at the end of the execution of line~\ref{LE:move} of Algorithm~\ref{alg:LE}.

If no iteration of $W1$ has been made before $I$ in $\mathcal{E}$, then at the beginning of $I$, $\mathcal{T}$ is either empty or contains only node $L[i]$ in view of lines~\ref{LE:if} to~\ref{LE:condition} of Algorithm~\ref{alg:LE}. This is still true after the pruning operation in $I$ and thus at the beginning of the execution of function {\tt Explore}$(l,bound)$ in $I$. This is a contradiction. Hence, at least one iteration of $W1$ has been made before $I$ in $\mathcal{E}$. Denote by $I'$ the iteration preceding $I$. If $I'$ occurs during $LE_i$, we know that property $P(i)$ is satisfied at the beginning of the execution of function {\tt Explore}$(l,bound)$ in $I'$. At the end of this execution, variable $bound$ is necessarily non-negative or otherwise we get a contradiction with the existence of $I$. By Lemma~\ref{lem:CDFX2} and Algorithm~\ref{alg:EX}, this implies that for each tree $Tr$ added into $\mathcal{T}$ during this execution of function {\tt Explore}, $|Tr|<\lfloor\frac{l}{2}\rfloor$ only if all the nodes of $Tr$ are complete in $\mathcal{M}$ at the end of the execution of function {\tt Explore}. Hence, in view of lines~\ref{clean1} to~\ref{clean2} of Algorithm~\ref{alg:LE}, property $P(i)$ is still true at the beginning of $I$. This implies that property $P(i)$ is also true at the beginning of the execution of function {\tt Explore}$(l,bound)$ in $I$ in view of Algorithm~\ref{alg:PR}. Indeed, via the pruning operation {\tt Prune}$(l)$ occuring in $I$, $T_u$ remains unchanged if $\epsilon_{T_u}(u)< \lfloor \frac{l}{4}\rfloor$, and it cannot be split into trees of size less than $\lfloor \frac{l}{8}\rfloor$ otherwise. This is again a contradiction. Consequently, $I'$ is the last iteration of $W1$ in $LE_k$ for some $1\leq k<i$ and $I$ is the first iteration of $W1$ in $LE_i$. By assumption, property $P(k)$ is satisfied at the beginning of the execution of function {\tt Explore}$(l,bound)$ in $I'$ and, using the same arguments as above, we know that property $P(k)$ is still true at the end of $I'$. Note that at the end of $I'$, node $L[k]$ is necessarily complete. Indeed, otherwise we get a contradiction with the fact that $I'$ is the last iteration of $W1$ in $LE_k$ in view of Claim~\ref{claim:GE2} and line~\ref{LE:condition} of Algorithm~\ref{alg:LE}. We can state that node $L[k]$ belongs to a tree containing a node different from $L[k]$ at the the end of $I'$, because otherwise it could not be in $\mathcal{T}$ at this time in view of line~\ref{clean1} of Algorithm~\ref{alg:LE}. As a result, the size of every tree of $\mathcal{T}$ is at least $\lfloor \frac{l}{8}\rfloor$ at the end of $I'$. Moreover, from the end of $I'$ to the beginning of $I$, $\mathcal{T}$ is subject to no change except the possible insertion of the one-node tree $L[i]$, in view of lines~\ref{LE:if} to~\ref{LE:condition} of Algorithm~\ref{alg:LE}. Consequently, property $P(i)$ is satisfied at the beginning of $I$. As explained above, this implies that property $P(i)$ is also true at the beginning of the execution of function {\tt Explore}$(l,bound)$ in $I$, which is again a contradiction and proves the claim.
\end{proofclaim}

Now, consider an iteration $I$ of $W1$ in $LE_i$ during which at least one blue or red move is made. At the beginning of function {\tt Explore}$(l,bound)$ in $I$, we know that the agent occupies an incomplete node $u$ belonging to a tree $Tr$ of $\mathcal{T}$. Moreover, the size of $Tr$ is at least $\lfloor \frac{l}{8}\rfloor$. Indeed, if it is not the case, $Tr$ is then the one-node tree $L[i]$ in view of Claim~\ref{claim:size}: thus the agent does not make at least one blue or red move during $I$, which is a contradiction with the definition of $I$. The number of blue moves (resp. red moves) in $I$ is at most $2|Tr|$ (resp. at most $l$), which gives a total cost of at most $10(|Tr|+1)$ for these two kinds of moves in $I$. The edges of $Tr$ originally come from trees returned during some previous calls to {\tt CDFS} and thus they all belong to $H$ in view of Lemma~\ref{lem:CDFX2} and the definition of $H$. First suppose that the execution of {\tt Explore}$(l,bound)$ returns a non-negative integer. In view of Lemma~\ref{lem:CDFX2} and Algorithm~\ref{alg:EX}, we know that all the nodes of $Tr$ (including node $u$) are complete in $\mathcal{M}$ at the end of it. If $|Tr|>0$, the number of blue or red moves in $I$ is at most $20|Tr|$ and they can be associated to $|Tr|$ edges of $H$ that will be removed from $\mathcal{T}$ via the execution of line~\ref{clean1} of Algorithm~\ref{alg:LE} in $I$: after that, these removed edges will never appear anymore in $\mathcal{T}$ during $\mathcal{E}$ by Claim~\ref{claim:remove}. If $|Tr|=0$, the number of blue or red moves in $I$ is at most $10$. These moves can be associated to node $u$ of $Tr$ that becomes complete in $\mathcal{M}$ during a call to {\tt CDFS} made in the execution of {\tt Explore}, and that remains so till the end of $\mathcal{E}$. Since $u$ is visited during a call to {\tt CDFS}, it necessarily belongs to $H$. Now, suppose that the execution of {\tt Explore}$(l,bound)$ returns a negative integer in $I$. This means that $I$ is the last execution of $W1$ in $\mathcal{E}$ and the number of blue and red moves during this last execution of $W1$ is at most $10(|H|+1)$ as $|Tr|\leq |H|$. From the above arguments, it follows that the number of blue or red moves in $\mathcal{E}$ is at most $30|H|+10z+10$ where $z$ is the number of nodes in $H$. Since $H$ is not necessarily a connected graph, it is not straightforward that $z$ can be upper bounded by a linear function in $|H|$. Hence, we need this last claim.

\begin{claim}
$z$ is at most $2|H|+1$.
\end{claim}

\begin{proofclaim}
Suppose by contradiction that $z\geq 2|H|+2$. This means that there exist two nodes $u$ and $u'$ of $H$ that have no incident edge in $H$. Recall that $H$ is the graph consisting of the nodes and edges that have been explored during the executions of {\tt CDFS} that are triggered throughout the calls to function {\tt Explore}$(l,bound)$. Hence, node $u$ (resp. $u'$) can belong to $H$ only if one of these executions of {\tt CDFS} has been launched from node $u$ (resp. $u'$). Denote by $X$ (resp. $X'$) this execution of {\tt CDFS} and suppose without loss of generality that $X$ occurs before $X'$. By Algorithm~\ref{alg:EX}, at the beginning of $X$ node $u$ is incomplete in $\mathcal{M}$. If at the end of $X$, node $u$ is still incomplete in $\mathcal{M}$, then, by Lemma~\ref{lem:CDFX2}, the value of variable $bound$ of Algorithm~\ref{alg:LE} is $-1$ at the end of $X$. Hence $X'$ cannot exist, which is a contradiction. As a result, node $u$ becomes complete in $\mathcal{M}$ during $X$. This means that at least one edge $e$ incident to $u$ is visited during $X$ and thus $e$ belongs to $H$. This contradicts the definition of node $u$ and proves the claim.
\end{proofclaim}

In view of the above claim, the number of blue or red moves in $\mathcal{E}$ is at most $50|H|+20$, which is $\mathcal{O}(m)$ by Claim~\ref{claim:GE4}. This closes the analysis of the case where $l\geq 2$ and $m\ne\perp$, and thus completes the proof of the lemma.
\end{proof}

Below is the lemma establishing the properties of procedure ${\tt Search}(x)$ that will be used to show the main theorem of this section. The proof of this lemma relies on Lemma~\ref{lem:GE2}.

\begin{lemma}
\label{lem:search1}
Consider an execution $\mathcal{E}$ of procedure ${\tt Search}(x)$ from the source node $s$, for any real constant $x>0$. Assume that $M_1(\mathcal{E})=B_f(\mathcal{G},s)$ for some integer $f\geq0$. 
\begin{itemize}
\item The execution terminates at node $s$ and during the execution the agent always knows a path in $\mathcal{G}$ of length at most $\max\{f+1,\lfloor (1+x)f \rfloor\}$ from node $s$ to its current node.
\item There exists an integer $f'>f$ such that $M_2(\mathcal{E})=B_{f'}(\mathcal{G},s)$ and at least one of the following properties holds:
\begin{enumerate}
\item The cost of $\mathcal{E}$ is $\mathcal{O}(e(f+1))$ and $xf<3$.
\item The cost of $\mathcal{E}$ is $\mathcal{O}(e(f))$ and $f'>(1+\frac{x}{3})f$.
\item The cost of $\mathcal{E}$ is $\mathcal{O}(e(f)\log (f+2))$ and $e(f'+1)\geq 2e(f)$.
\item The cost of $\mathcal{E}$ is $\mathcal{O}(e(f+1))$ and $e(f+1)\geq 2e(f)$.
\end{enumerate}
\end{itemize}
\end{lemma}
\begin{proof}
The execution of procedure ${\tt Search}(x)$ from node $s$ can be viewed as a sequence $S$ of consecutive executions $\mathcal{E}_0, \mathcal{E}_1, \mathcal{E}_2,\ldots$ in which $\mathcal{E}_i$ corresponds to the $i$th execution of the while loop of Algorithm~\ref{alg:SE} if $i\geq 1$, and to the instructions before the while loop of Algorithm~\ref{alg:SE} otherwise. The length of $S$, \ie the number of executions $\mathcal{E}_i$ in $S$, will be denoted by $|S|$ (we show below that $|S|$ is finite). We will often discuss the values of three specific variables of Algorithm~\ref{alg:SE} that are $l$, $floor$ and $ceil$: in the sequel, the values of $l$, $floor$ and $ceil$ at the end of $\mathcal{E}_k$ will be respectively denoted by $l_k$, $floor_k$ and $ceil_k$. Observe that the existence of $\mathcal{E}_{i+1}$ implies that $l_i\geq1$ according to the while loop of Algorithm~\ref{alg:SE}. Also observe that the value of variable $floor$ is always a non-negative integer that never decreases, and the second parameter of each call to function ${\tt GlobalExpansion}$ is always a positive integer (cf. condition of the while loop of Algorithm~\ref{alg:SE}) except for the first call in which it is $\perp$. These observations, which condition the validity of several subsequent arguments, must be kept in mind when reading the proof as they will not always be repeated in order to lighten the text.

We start by showing two claims.

\begin{claim}
\label{searchcl0}
For every $0\leq k\leq |S|-1$, $\mathcal{E}_k$ starts and ends at node $s$, $M_2(\mathcal{E}_k)=B_{floor_k}(\mathcal{G},s)$ and $floor_k+2l_k-1\leq ceil_k\leq \lfloor (1+x)f \rfloor$.
\end{claim}

\begin{proofclaim}
We prove the claim by induction on $k$ and we begin with the base case $k=0$. Execution $\mathcal{E}_0$ essentially consists of a call to ${\tt GlobalExpansion}(1,\perp)$ from node $s$. By assumption we have $M_1(\mathcal{E}_0)=B_f(\mathcal{G},s)$. Hence, according to Lemma~\ref{lem:GE2}, $\mathcal{E}_0$ terminates at node $s$ and $M_2(\mathcal{E}_0)=B_{floor_{0}}(\mathcal{G},s)$. Moreover, in view of lines~\ref{assign0} and~\ref{assign1} of Algorithm~\ref{alg:SE}, we have $floor_0=f+1$ and $floor_0+2l_0-1\leq ceil_0\leq \lfloor (1+x)f \rfloor$, which concludes the base case $k=0$.

Now, assume that there is an integer $0\leq k\leq|S|-2$ such that $\mathcal{E}_k$ terminates at node $s$, $M_2(\mathcal{E}_k)=B_{floor_k}(\mathcal{G},s)$ and $floor_k+2l_k-1\leq ceil_k\leq \lfloor (1+x)f \rfloor$. We prove below that these properties also hold for $k+1$. 

Execution $\mathcal{E}_{k+1}$ essentially consists of a call to ${\tt GlobalExpansion}(l_k,m)$. By the inductive hypothesis, $\mathcal{E}_{k+1}$ starts at node $s$ and $M_1(\mathcal{E}_{k+1})=B_{floor_k}(\mathcal{G},s)$. It then follows from Lemma~\ref{lem:GE2} that $\mathcal{E}_{k+1}$ terminates at node $s$, and $B_{floor_{k}+l_k}(\mathcal{G},s)\subseteq M_2(\mathcal{E}_{k+1})$ if ${\tt GlobalExpansion}(l_k,m)$ returns true, $B_{floor_{k}}(\mathcal{G},s)\subseteq M_2(\mathcal{E}_{k+1})$ otherwise. Note that according to lines~\ref{condassign2} to~\ref{fincond} of Algorithm~\ref{alg:SE}, $floor_{k+1}=floor_k+l_k$ if ${\tt GlobalExpansion}(l_k,m)$ returns true, $floor_{k+1}=floor_k$ otherwise. Thus, in view of line~\ref{resizeball} of Algorithm~\ref{alg:SE}, we have $M_2(\mathcal{E}_{k+1})=B_{floor_{k+1}}(\mathcal{G},s)$. Finally, it remains to prove that $floor_{k+1}+2l_{k+1}-1\leq ceil_{k+1}\leq \lfloor (1+x)f \rfloor$. If ${\tt GlobalExpansion}(l_k,m)$ returns true,  then $ceil_{k+1}=ceil_k$ and $l_{k+1}= \lfloor \frac{ceil_{k+1}-floor_{k+1}}{2} \rfloor$ (cf. line~\ref{assign2} of Algorithm~\ref{alg:SE}), which implies, in view of the inductive hypothesis, that $floor_{k+1}+2l_{k+1}-1\leq ceil_{k+1}\leq \lfloor (1+x)f \rfloor$. If ${\tt GlobalExpansion}(l_k,m)$ returns false, then $ceil_{k+1}=floor_k+2l_k-1$, $floor_{k+1}=floor_k$ (cf. line~\ref{assign3} of Algorithm~\ref{alg:SE}) and $l_{k+1}=\lfloor \frac{l_k}{2}\rfloor$ which also implies, in view of the inductive hypothesis, that $floor_{k+1}+2l_{k+1}-1\leq ceil_{k+1}\leq \lfloor (1+x)f \rfloor$. This concludes the inductive proof of the claim.
\end{proofclaim}

\begin{claim}
\label{searchcl4}
$|S|$ is in $\mathcal{O}(\log (f+2))$.
\end{claim}

\begin{proofclaim}
If $|S|=1$, the claim trivially holds. Hence, suppose that $|S|\geq2$ and fix any integer $1\leq i \leq |S|-1$. We show below that there is an integer $c\leq 3$ such that $l_{i+c}\leq \frac{7l_i}{8}$ or $i\geq|S|-6$. This is enough to prove the claim. Indeed, the above property implies that $|S|$ is $\mathcal{O}(\log (f+2))$ because $l_1$ is at most linear in $f$ by Claim~\ref{searchcl0} and because we exit the while loop when the value of variable $l$ becomes less than $1$. We consider two cases.

\begin{itemize}
\item {\it Case~1: the execution of ${\tt GlobalExpansion}(l_{i-1},m)$ returns true in $\mathcal{E}_i$ and $i\leq|S|-3$}. (Note that executions $\mathcal{E}_{i+1}$ and $\mathcal{E}_{i+2}$ exist as $i\leq|S|-3$). By line~\ref{assign2} of Algorithm~\ref{alg:SE},  we get $l_{i}= \lfloor\frac{ ceil_i-floor_i }{2}\rfloor$. Let $i\leq j\leq i+2$ be the largest integer such that the return value of ${\tt GlobalExpansion}$ is true from the execution $\mathcal{E}_i$ to $\mathcal{E}_{j}$ included. In view of line~\ref{assign2} of Algorithm~\ref{alg:SE} and the fact that $l_{i}= \lfloor\frac{ ceil_i-floor_i }{2}\rfloor$, if $j=i+2$ then $l_{i+2}$ cannot be more than $\frac{l_{i}}{2}$. Still from the same arguments, we know that variable $l$ never increases from the end of $\mathcal{E}_{i}$ to the end of $\mathcal{E}_{j}$, and thus $l_{j}\leq l_i$. Hence, if $j<i+2$, the return value of ${\tt GlobalExpansion}$ is false in $\mathcal{E}_{j+1}$, which implies, according to line~\ref{assign3} of Algorithm~\ref{alg:SE}, that $l_{j+1}\leq \frac{l_j}{2}\leq \frac{l_i}{2}$. Consequently, in the first case, $l_{i+1}$ or $l_{i+2}$ is at most $\frac{l_i}{2}$.


\item {\it Case~2: the execution of ${\tt GlobalExpansion}(l_{i-1},m)$ returns false in $\mathcal{E}_i$ and $i\leq|S|-4$}. (Note that executions $\mathcal{E}_{i+1}$ to $\mathcal{E}_{i+3}$ exist as $i\leq|S|-4$). By line~\ref{assign3} of Algorithm~\ref{alg:SE}, we get $l_{i}=\lfloor\frac{l_{i-1}}{2}\rfloor$ and $ceil_i-floor_i=2l_{i-1}-1$. If the execution of ${\tt GlobalExpansion}$ returns also false in $\mathcal{E}_{i+1}$, then we have $l_{i+1}\leq \frac{l_i}{2}$. So, suppose that the execution of ${\tt GlobalExpansion}$ returns true in $\mathcal{E}_{i+1}$. We then have $l_{i+1}=\lfloor\frac{ceil_{i+1}-floor_{i+1}}{2}\rfloor=\lfloor\frac{ceil_i-floor_i-l_i}{2}\rfloor=\lfloor\frac{2l_{i-1}-1-l_i}{2}\rfloor$. Furthermore, $i+1\leq|S|-3$. Thus, using the same reasoning as in Case~1 (but replacing $i$ by $i+1$), we have $l_{i+2}\leq \frac{l_{i+1}}{2}$ or $l_{i+3}\leq \frac{l_{i+1}}{2}$. We consider three subcases.

\begin{itemize}
\item{\it Subcase~2.1: $l_{i-1}$ is even.} We have $l_{i+1}=\lfloor\frac{2l_{i-1}-1-l_i}{2}\rfloor=\lfloor \frac{3l_i-1}{2}\rfloor\leq\lfloor \frac{3l_i}{2}\rfloor$. Thus, $l_{i+2}$ or $l_{i+3}$ is at most $\frac{3l_i}{4}$.

\item{\it Subcase~2.2: $l_{i-1}$ is odd and $l_i\geq 2$.} We have $l_{i+1}=\lfloor\frac{2l_{i-1}-1-l_i}{2}\rfloor=\lfloor \frac{3l_i+1}{2}\rfloor\leq\lfloor \frac{7l_i}{4}\rfloor$. Thus, $l_{i+2}$ or $l_{i+3}$ is at most $\frac{7l_i}{8}$.

\item{\it Subcase~2.3: $l_{i-1}$ is odd, $l_i=1$ and $i\leq |S|-7$.} We have $l_{i+1}=\lfloor\frac{2l_{i-1}-1-l_i}{2}\rfloor=\lfloor \frac{3l_i+1}{2}\rfloor=2$. Note that $i+2\leq |S|-5$ and $l_{i+1}$ is even. Hence, using the same reasoning as from Case~1 to Subcase~2.1 (but replacing $i$ by $i+2$), we know that $l_{i+3}$, $l_{i+4}$ or $l_{i+5}$ is at most $\frac{3l_{i+2}}{4}$. Also note that $l_{i-1}=3$ as 
$l_{i}=\lfloor\frac{l_{i-1}}{2}\rfloor$. Thus, $ceil_i=floor_i+5$ since $ceil_i-floor_i=2l_{i-1}-1$. This implies that $ceil_{i+1}=floor_{i+1}+4$ in view of line~\ref{assign2} of Algorithm~\ref{alg:SE}. Consequently, if the execution of ${\tt GlobalExpansion}$ returns true in $\mathcal{E}_{i+2}$, we have $ceil_{i+2}-floor_{i+2}=ceil_{i+1}-floor_{i+1}-l_{i+1}= 2$, and thus $l_{i+2}=1$. Besides, if the execution of ${\tt GlobalExpansion}$ returns false in $\mathcal{E}_{i+2}$, we immediately have $l_{i+2}=\lfloor\frac{l_{i+1}}{2} \rfloor=1$. As a result, $l_{i+3}$, $l_{i+4}$ or $l_{i+5}$ is less than $1$, which means that $i\geq |S|-6$ and contradicts the assumption that $i\leq |S|-7$.
\end{itemize}
\end{itemize}
To summarize, we have shown that $i\geq |S|-6$ or there is an integer $c\leq 3$ such that $l_{i+c}\leq \frac{7l_i}{8}$. This concludes the proof of the claim.
\end{proofclaim}

According to Claims~\ref{searchcl0} and~\ref{searchcl4}, $S$ terminates after $\mathcal{O}(\log (f+2))$ iterations of the while loop, and $M_2(\mathcal{E})=M_2(\mathcal{E}_{|S|-1})=B_{floor_{|S|-1}}(\mathcal{G},s)$ with $floor_{|S|-1}\geq floor_0>f$. By assumption, $\mathcal{E}_0$ starts at node $s$ and $M_1(\mathcal{E}_0)=B_{f}(\mathcal{G},s)$. It then follows from Lemma~\ref{lem:GE2} that the agent always knows during $\mathcal{E}_0$ a path in $\mathcal{G}$ of length at most $f+1$ from $s$ to its current node. By Claim~\ref{searchcl0}, for every $1\leq k\leq |S|-1$, $\mathcal{E}_k$ starts at node $s$ and $M_1(\mathcal{E}_k)=B_{floor_k}(\mathcal{G},s)$. Hence, from Lemma~\ref{lem:GE2} and Claim~\ref{searchcl0},  it follows that for every $1\leq k\leq |S|-1$, the agent always knows during $\mathcal{E}_k$ a path in $\mathcal{G}$ of length at most $floor_{k-1}+2l_{k-1}-1\leq ceil_{k-1}\leq \lfloor (1+x)f \rfloor$ from $s$ to its current node. This completes the proof of the first part of the lemma.

Now, we prove the second part. To do so, it is enough to show that one of the four properties of the second part is satisfied with $f'=floor_{|S|-1}$ as $M_2(\mathcal{E})=B_{floor_{|S|-1}}(\mathcal{G},s)$ and $floor_{|S|-1}>f$.

Since $\mathcal{E}_0$ starts at node $s$ and $M_1(\mathcal{E}_0)=B_{f}(\mathcal{G},s)$, it follows from Lemma~\ref{lem:GE2} that the cost of $\mathcal{E}_0$ is $\mathcal{O}(e(f+1))$. After $\mathcal{E}_0$, the condition of the while loop in Algorithm~\ref{alg:SE} evaluates to false (and thus $|S|=1$) if and only if $l_0<1$ or $|M_2(\mathcal{E}_0)|\geq 2e(f)$ \ie if and only if $xf<3$ or $e(floor_0)\geq 2e(f)$. As a result, if $|S|=1$, 
$f'=floor_{|S|-1}=floor_0=f+1$ and the first or the fourth property are satisfied.

So, assume that $|S|\geq2$. From the above explanations we have $xf\geq 3$.

In $\mathcal{E}_1$, the agent executes ${\tt GlobalExpansion}(l_0,m)$ where $m$ is $e(f)$ in view of line~\ref{SEf} of Algorithm~\ref{alg:SE}. We stated earlier that $\mathcal{E}_0$ is in $\mathcal{O}(e(f+1))$. Besides, $M_2(\mathcal{E}_0)=B_{f+1}(\mathcal{G},s)$ by Claim~\ref{searchcl0}. Consequently, the existence of $\mathcal{E}_1$ and the condition of the while loop of Algorithm~\ref{alg:SE} imply that $|M_2(\mathcal{E}_0)|=e(f+1)<2m=2e(f)$, and we get the following claim.

\begin{claim}
\label{searchcl2}
If $|S|\geq2$, the cost of $\mathcal{E}_0$ is $\mathcal{O}(e(f))$.
\end{claim}

By Lemma~\ref{lem:GE2} and Claim~\ref{searchcl0}, $\mathcal{E}_1$ terminates and its cost is $\mathcal{O}(m)$, which is $\mathcal{O}(e(f))$ as $m=e(f)$, $M_1(\mathcal{E})=B_{f+1}(\mathcal{G},s)$ and $e(f+1)<2m$. If ${\tt GlobalExpansion}(l_0,m)$ returns true in $\mathcal{E}_1$, we know that $|S|=2$ in view of the condition of the while loop of Algorithm~\ref{alg:SE}. Since the cost of $\mathcal{E}_0$ is $\mathcal{O}(e(f))$ (cf. Claim~\ref{searchcl2}), the total cost of $S$ (and thus of $\mathcal{E}$) is $\mathcal{O}(e(f))$. Moreover, $floor_1=floor_0+l_0= f+1+\lfloor \frac{\lfloor xf\rfloor-1}{2} \rfloor$ which is more than $(1+\frac{x}{3})f$, as $xf\geq 3$. Hence, at the end of $S$ the second property is satisfied with $f'=floor_{|S|-1}=floor_1$.

Now, consider the case where ${\tt GlobalExpansion}(l_0,m)$ returns false in $\mathcal{E}_1$. To deal with this case, we need to prove the following claim.

\begin{claim}
\label{searchcl3}
For every $1\leq k\leq |S|-1$, $e(ceil_k)> 2e(f)$.
\end{claim}

\begin{proofclaim}
We prove the claim by induction on $k$ and we start with the base case $k=1$. By Claim~\ref{searchcl0}, $M_1(\mathcal{E}_1)=B_{floor_0}(\mathcal{G},s)$ and $\mathcal{E}_1$ starts at node $s$. From this, Lemma~\ref{lem:GE2} and the fact that ${\tt GlobalExpansion}(l_0,m)$ returns false, we have $e(floor_0+2l_{0}-1)> e(floor_0)+m$. Since $m=e(f)$, $ceil_1=floor_0+2l_{0}-1$ (cf. line~\ref{assign3} of Algorithm~\ref{alg:SE}) and $e(floor_0)\geq e(f)$ as $floor_0=f+1$, we get $e(ceil_1)> 2e(f)$, which concludes the base case $k=1$.

Now, assume that there is an integer $1\leq k\leq|S|-2$ such that $e(ceil_k)> 2e(f)$. We prove that $e(ceil_{k+1})> 2e(f)$. 

By Claim~\ref{searchcl0}, $\mathcal{E}_{k+1}$ terminates. If ${\tt GlobalExpansion}(l_k,m)$ returns true then $ceil_{k+1}=ceil_k$, which implies that $e(ceil_{k+1})>2e(f)$ as $e(ceil_k)>2e(f)$ by the inductive hypothesis. Otherwise, ${\tt GlobalExpansion}(l_k,m)$ returns false: in view of Lemma~\ref{lem:GE2} and Claim~\ref{searchcl0}, we then have $e(floor_{k}+2l_{k}-1)> e(floor_k)+m$, which implies $e(ceil_{k+1})>2e(f)$ as $ceil_{k+1}=floor_{k}+2l_{k}-1$ (cf. line~\ref{assign3} of Algorithm~\ref{alg:SE}) and $floor_k\geq f+1$. This completes the inductive proof of the claim.
\end{proofclaim}

With the above claim, we are ready to conclude the case where ${\tt GlobalExpansion}(l_0,m)$ returns false in $\mathcal{E}_1$. For every $1\leq k\leq |S|-1$, the cost of $\mathcal{E}_k$ is equal to the cost of the execution of ${\tt GlobalExpansion}(l_{k-1},m)$, which is $\mathcal{O}(e(floor_{k-1})+m)$ by Lemma~\ref{lem:GE2} and Claim~\ref{searchcl0}. Besides, $m=e(f)$ and $e(floor_{k-1})<2m$ for every $1\leq k\leq |S|-1$ (by the condition of the while loop of Algorithm~\ref{alg:SE}). Thus, from Claims~\ref{searchcl4} and~\ref{searchcl2}, the cost of $S$ is $\mathcal{O}(e(f)\log (f+2))$. Since $\mathcal{E}_{|S|-1}$ corresponds to the last iteration of the while loop, we get $l_{|S|-1}<1$ or $e(floor_{|S|-1})\geq 2m$ ($M_2(\mathcal{E}_{|S|-1})=B_{floor_{|S|-1}}(\mathcal{G},s)$ by Claim~\ref{searchcl0}). Note that from lines~\ref{condassign2} to~\ref{assign3}, it follows that $ceil_{|S|-1}-floor_{|S|-1}=2l_{|S|-2}-1$ and $l_{|S|-1}=\lfloor \frac{l_{S-2}}{2}\rfloor$, or $ceil_{|S|-1}-floor_{|S|-1}\leq 2l_{S-1}+1$. Hence, if $l_{|S|-1}<1$, we  have $ceil_{|S|-1}\leq floor_{|S|-1}+1$, which implies that $e(floor_{|S|-1}+1)> 2e(f)$ by Claim~\ref{searchcl3}. Furthermore, if $l_{|S|-1}\geq1$ then $e(floor_{|S|-1})\geq 2m$ and we have $e(floor_{|S|-1}+1)\geq 2e(f)$ because $m=e(f)$. As a result, the third property of the lemma is satisfied with $f'=floor_{|S|-1}$. This completes the analysis of the case where ${\tt GlobalExpansion}(l_0,m)$ returns false in $\mathcal{E}_1$, and thus concludes the proof of this lemma.
 \end{proof}



If we put aside the initial assignments of lines~\ref{TH1} and~\ref{THadd} in Algorithm~\ref{alg:TH}, the execution of procedure ${\tt TreasureHunt}(x)$ from the source node $s$ in $\mathcal{G}$ can be viewed as a sequence of consecutive executions of procedure ${\tt Search}(x)$: the i$th$ execution of ${\tt Search}(x)$ in this sequence will be denoted by $S_i$.

The following lemma is a small technical observation concerning the execution of ${\tt TreasureHunt}(x)$ from the source node $s$. Since, at the beginning of this execution, variable $\mathcal{M}$ is equal to $B_0(\mathcal{G},s)$, the lemma can be easily proved by induction on $i$ using Lemma~\ref{lem:search1}.

\begin{lemma}
\label{claim:theo0}
Consider an execution of procedure ${\tt TreasureHunt}(x)$ from the source node $s$, for any real constant $x>0$. For every integer $i\geq1$, $S_i$ starts and ends at node $s$, and there are two integers $f_{i+1}>f_i\geq i-1$ such that $M_1(S_i)=B_{f_i}(\mathcal{G},s)$ and $M_2(S_i)=B_{f_{i+1}}(\mathcal{G},s)$.
\end{lemma}



We are now ready to prove the main result of this section that is stated in the following theorem.

\begin{theorem}
\label{theo:final0}
Consider a graph $\mathcal{G}$ of unknown radius $r$ in which a treasure is located at an unknown distance at most $1<d\leq r$ from the starting node $s$ of an agent. For any real constant $x>0$,  procedure ${\tt TreasureHunt}(x)$ allows the agent to find the treasure at cost $\mathcal{O}(e(d)\log d)$.
\end{theorem}
\begin{proof}
Let $\mathcal{S}=(S_1,S_2,S_3,\ldots)$ be the sequence of consecutive executions of ${\tt Search}(x)$ resulting from the execution of procedure ${\tt TreasureHunt}(x)$ from node $s$. The length of $\mathcal{S}$, \ie the number of executions of ${\tt Search}(x)$ in $\mathcal{S}$, will be denoted by $|\mathcal{S}|$ (we show below that $|\mathcal{S}|$ is upper bounded by $d$). According to the procedure, the sequence $\mathcal{S}$ is interrupted as soon as the treasure is found during some execution $S_i$, which means that $S_i$ is interrupted before its ``natural end''. Nonetheless, to simplify the subsequent discussions, we will suppose in this proof that, instead of stopping as soon as the treasure is found, the agent stops at the end of the first execution of ${\tt Search}(x)$ during which it has found the treasure. Hence, all the executions of ${\tt Search}(x)$ in $\mathcal{S}$ are complete.

Variable $\mathcal{M}$ can be modified only through the instructions of line~\ref{THadd} of Algorithm~\ref{alg:TH}, line~\ref{resizeball} of Algorithm~\ref{alg:SE} and lines~\ref{CDFSadd1} to~\ref{CDFSadd2} of Algorithm~\ref{alg:CDFS}. In view of these instructions, the value of $\mathcal{M}$ is always a subgraph of $\mathcal{G}$ whose nodes and edges have been explored by the agent. Thus, we know that for every integer $1\leq i\leq |\mathcal{S}|$, $B_{d}(\mathcal{G},s) \nsubseteq M_1(S_i)$. Indeed, if it was not the case for some $i$, we would get a contradiction with the existence of $S_i$, as this would mean that the treasure has been found before the start of $S_i$. Note that $|\mathcal{S}|\leq d$ because if $S_{d+1}$ existed, we would have $B_{d}(\mathcal{G},s)\subseteq M_1(S_{d+1})$ in view of Lemma~\ref{claim:theo0}.

For every $1\leq i\leq |\mathcal{S}|$, we denote by $f_i$ the integer such that $M_1(S_i)=B_{f_i}(G,s)$. The existence and the unicity of each of these integers are guaranteed by Lemma~\ref{claim:theo0} and the fact that $B_{d}(\mathcal{G},s) \nsubseteq M_1(S_i)$. Still by Lemma~\ref{claim:theo0}, we know that for every $1\leq i < |\mathcal{S}|$, $0\leq f_i<f_{i+1}$.

From all the above explanations, we get the following claim.

\begin{claim}
\label{claim:theo}
$|\mathcal{S}|\leq d$, the treasure is found by the end of $\mathcal{S}$, and for every $1\leq i < |\mathcal{S}|$, $0\leq f_i<f_{i+1}<d$.
\end{claim}

Since the agent ends up finding the treasure, it remains to discuss the incurred cost.

Note that in view of Lemma~\ref{claim:theo0}, there is an integer $f_{|\mathcal{S}|+1}$ such that $M_2(S_{|\mathcal{S}|})=B_{f_{|\mathcal{S}|+1}}(G,s)$. However, we may have several candidates for $f_{|\mathcal{S}|+1}$ because $M_2(S_{|\mathcal{S}|})$ may be $\mathcal{G}$. In the rest of this proof, $f_{|\mathcal{S}|+1}$ is choosen as the smallest integer that satisfies one of the four properties of the second part of Lemma~\ref{lem:search1} (with $f'=f_{|\mathcal{S}|+1}$, and $f=f_{|\mathcal{S}|}$ since $M_1(S_{|\mathcal{S}|})=B_{f_{|\mathcal{S}|}}(G,s)$). 


In view of Lemma~\ref{lem:search1}, we know that for each execution $S_i$ in $\mathcal{S}$, at least one of the four properties of the second part of Lemma~\ref{lem:search1} is satisfied with $f=f_i$ and $f'=f_{i+1}$: the execution is then said to be of type $1\leq j\leq 4$ if it satisfies the $j$th property. The cost of $\mathcal{S}$ is upper bounded by $C_1+C_2+C_3+C_4$ where $C_j$ is the total cost of the executions of type~$j$ in $\mathcal{S}$. So, to prove the theorem, it is enough to prove that $C_1$, $C_2$, $C_3$ and $C_4$ all belong to $\mathcal{O}(e(d)\log d)$. This is the purpose of the rest of this proof. 

First, consider the case of $C_1$. We know from Lemma~\ref{claim:theo0} that if $i\geq\lceil \frac{3}{x}\rceil+1$, then $f_i\geq\lceil \frac{3}{x}\rceil$, which means that $f_ix\geq 3$ and $S_i$ cannot be of type~$1$ according to Lemma~\ref{lem:search1}. Thus, the number of executions of type~$1$ in $\mathcal{S}$ is at most $\lceil \frac{3}{x}\rceil$. Moreover, if $S_i$ is of type $1$, its cost is $\mathcal{O}(e(f_i+1))$ and $f_i<d$ in view of Lemma~\ref{lem:search1} and Claim~\ref{claim:theo}. Hence, $C_1$ is in $\mathcal{O}(e(d))$, which is $\mathcal{O}(e(d)\log d)$.

Now, consider the case of $C_2$. If $S_i$ is of type $2$, it follows from Lemma~\ref{lem:search1} that the cost of $S_i$ is $\mathcal{O}(e(f_i))$ and $f_{i+1}>(1+\frac{x}{3})f_i$. From Claim~\ref{claim:theo}, we must have $f_i<d$. Moreover, $f_1=0$ and by Lemma~\ref{claim:theo0}  $f_i>f_{i-1}$ for every $2\leq i\leq |\mathcal{S}|$. Hence, the number of executions of type~$2$ in $\mathcal{S}$ is $\mathcal{O}(\log d)$ and the cost of each of them is $\mathcal{O}(e(d))$. This implies that $C_2$ is in $\mathcal{O}(e(d)\log d)$.

Let us turn attention to the case of $C_3$. From Lemma~\ref{lem:search1}, there is a constant $c>0$ such that for every $1\leq i\leq |\mathcal{S}|$, we have the following: if $S_i$ is of type $3$, then $e(f_{i+1}+1)\geq 2e(f_i)$ and the cost of $S_i$ is at most $c \cdot e(f_i)\log(2+f_i)=h(f_i)$. Let $h(f_i)=c \cdot e(f_i)\log(2+f_i)$. Let $i^*$ be the smallest integer, if any, such that $S_{i^*}$ is of type $3$ and $i^*=i+2k$ for some integer $k\geq 1$. From the above explanations and from Claim~\ref{claim:theo}, it follows that, if $S_i$ is of type $3$, then $f_{i^*}\geq f_{i+2}\geq f_{i+1}+1>f_i$ and $e(f_{i^*})\geq e(f_{i+1}+1)\geq 2e(f_i)$, which means that $h(f_{i^*})\geq 2h(f_{i})$. Let $k$ (resp. $k'$) be the largest even (resp. odd) integer, if any, such that $S_k$ (resp. $S_{k'}$) is of type $3$. By the telescopic effect, the sum of the costs of the executions $S_i$ of type~3 such that $i$ is even (resp. odd) is at most $2h(f_k)$ (resp. $2h(f_{k'})$). Since by Claim~\ref{claim:theo}, $f_i<d$ for every $1\leq i\leq |\mathcal{S}|$, we know that $2h(f_k)$ as well as $2h(f_{k'})$ is at most $2c \cdot e(d)\log(1+d)$. Hence, $C_3$ is in $\mathcal{O}(e(d)\log d)$.



Using, as above, arguments based on the telescopic effect, we can show that $C_4$ is in $\mathcal{O}(e(f_k+1))$ where $k$ is the largest integer, if any, such that $S_k$ is of type $4$. Since $f_k<d$, $C_4$ is in $\mathcal{O}(e(d))$, which is $\mathcal{O}(e(d)\log d)$. This concludes the proof of the theorem.
\end{proof}

\section{Treasure hunt with restrictions}
\label{sec:rest}
Theorem \ref{theo:final0} holds for the task of treasure hunt without any restrictions on the moves of the agent, for all locally finite graphs, both finite and infinite.  In this section we show how {to modify our treasure hunt algorithm} to make it work under the fuel-restricted and the rope-restricted models for finite graphs.

{Strictly speaking, the fuel-restricted model was defined in  \cite{ABRS} assuming that both the constant $\alpha>0$ and the radius $r$ were known to the agent.} On the other hand, the rope-restricted model was defined in \cite{DKK} for any known constant $\alpha>0$  and for unknown radius $r$.  We will show that, for each of these restrictive models and for any known constant $\alpha>0$, {we can design a treasure hunt algorithm with} the promised efficiency even when $r$ is unknown. To this end, we need to modify the restriction of the fuel-restricted model from \cite{ABRS}, avoiding to reveal $r$ to the agent by showing it the size of the tank. We fix a positive constant $\alpha$, known to the agent, and we proceed as follows. For the restricted tank case from \cite{ABRS}, we assume that at any  visit of $s$ the agent can put as much fuel in the tank as it wants, but we show that if the (unknown) radius of the graph is $r$ then the tank is never filled to more than $B=2(1+\alpha)r$. The formalization of the rope-restricted model corresponds to its definition in \cite{DKK}. Recall that the agent is attached at $s$ by an infinitely extendible rope that it unwinds by a length 1 with every forward edge traversal and rewinds by a length of 1 with every backward edge traversal. Whenever the agent completely backtracks to $s$, the unwinded segment of the rope is of length 0.  We show that if the (unknown) radius of the graph is $r$ then the initial segment of the rope unwinded by the agent executing our algorithm will never be longer than $L=(1+\alpha)r$.

{The following theorem states that procedure ${\tt TreasureHunt}$ can be transformed into a procedure allowing the agent to find the treasure in the aforementioned restrictive models, without changing the asymptotic complexity.

\begin{theorem}
\label{theo:final}
{Consider a graph $\mathcal{G}$ of unknown radius $r$ in which a treasure is located at an unknown distance at most $1<d\leq r$ from the starting node $s$ of the agent. For any positive constant $\alpha$, procedure ${\tt TreasureHunt}(\frac{\alpha}{2})$ can be transformed into a procedure allowing the agent to find the treasure at cost $\mathcal{O}(e(d)\log d)$ in the rope-restricted model (resp. fuel-restricted model) without ever using a segment of the rope longer than $(1+\alpha)r$ (resp. without filling the tank to more than $2(1+\alpha)r$ at any visit of $s$).}
\end{theorem}

\begin{proof}
The execution of procedure ${\tt TreasureHunt}(\frac{\alpha}{2})$ from node $s$ corresponds to a sequence $S=(S_1,S_2,\ldots, S_{|S|})$ of executions of ${\tt Search}(\frac{\alpha}{2})$, in which the $|S|$th execution of ${\tt Search}(\frac{\alpha}{2})$ is interrupted prematurely because of the discovery of the treasure. 

We denote by $G_0$ the graph consisting only of node $s$, and for every $1\leq i\leq |S|$, we denote by $G_i$ the subgraph of $\mathcal{G}$ that has been explored from the beginning of $S_1$ to the end of $S_i$. For every $1\leq i\leq |S|$, the cost of $S_i$ will be denoted by $c_i$.


According to Lemma~\ref{claim:theo0}, for every $1\leq i\leq |S|$, $S_i$ starts and ends at node $s$ (except $S_{|S|}$ that ends at the node containing the treasure), there is an integer $f_i\geq0$ such that $M_1(S_i)=B_{f_i}(\mathcal{G},s)$ and if $i<|S|$, $M_2(S_i)=M_1(S_{i+1})$. Moreover, the value of $\mathcal{M}$ is always a subgraph of $\mathcal{G}$ whose nodes and edges have been all explored by the agent, and thus, for every $1\leq i \leq |S|$, $B_{f_i}(\mathcal{G},s)$ is a subgraph of $G_{i-1}$, $f_i$ is unique and $f_i<d$ (or otherwise the treasure would have been found before the start of ${S}_i$ which leads to a contradiction with the existence of this execution). Hence, from the fact that $d\leq r$, we get the following claim.

\begin{claim}
\label{claim:path}
For every $1\leq i\leq |S|$, $\max\{f_i+1,\lfloor (1+\alpha)f_i\rfloor\}\leq (1+\alpha)r$
\end{claim}


First, we describe a new algorithm $\mathcal{A}$ that permits to find the treasure, in the model without constraints, with asymptotically the same cost as that of  ${\tt TreasureHunt}(\frac{\alpha}{2})$. This new algorithm consists in executing ${\tt TreasureHunt}(\frac{\alpha}{2})$ with some changes in order to guarantee an extra property that will be important for our purpose. More precisely, an execution of $\mathcal{A}$ from node $s$ is a sequence of executions $(S'_1,S'_2,\ldots ,S'_{|S|})$ in which each $S'_i$ has cost $\mathcal{O}(c_i)$ and corresponds to an emulation of execution $S_i$. In particular, for every $1\leq i\leq |S|$, $S'_i$ starts and ends at node $s$ (except $S'_{|S|}$ that ends at the node containing the treasure), $M_1(S'_i)=M_1(S_i)=B_{f_i}(\mathcal{G},s)$, and at the end of $S'_i$, $G_{i}$ has been entirely explored. Obviously, all of this would not be interesting without the additional crucial property brought by $S'_i$ that will be called the \emph{frequent return property} and that is the following. Let $Sk$ be the stack initially empty in which we push (resp. pop) the last traversed edge if it corresponds to a forward (resp. backward) edge traversal. During $S'_i$, the size of $Sk$ is $0$ at least once during any block of $2\max\{f_i+1,\lfloor (1+\alpha)f_i\rfloor\}$ consecutive edge traversals, and is never greater than $\max\{f_i+1,\lfloor (1+\alpha)f_i\rfloor\}$. Moreover, at the beginning of $S'_i$, the size of $Sk$ is $0$, and if $i<|S|$, it is also $0$ at the end of $S'_i$.




Note that Algorithm $\mathcal{A}$ is a solution with the desired cost in the rope-restricted model, that will never use a segment of the rope longer than $(1+\alpha)r$, as for all $1\leq i \leq |S|$, we have $\max\{f_i+1,\lfloor (1+\alpha)f_i\rfloor\}\leq (1+\alpha)r$ according to Claim~\ref{claim:path}. By requiring the agent, each time the size of $Sk$ is $0$ in $S'_i$, to refuel its tank up to the limit of $2\cdot\max\{f_i+1,\lfloor (1+\alpha)f_i\rfloor\}$ (when the size of $Sk$ is $0$, the agent is at node $s$), we also get our objective with algorithm $\mathcal{A}$ in the fuel-restricted model, as the agent never runs out of fuel and $2\cdot\max\{f_i+1,\lfloor (1+\alpha)f_i\rfloor\}\leq 2(1+\alpha)r$.

Let us describe how we can construct our emulations while ensuring the features mentioned above. Consider the emulation $S'_i$ of $S_i$. Assume that at the beginning of $S'_i$, $G_{i-1}$ has been entirely explored, the size of $Sk$ is $0$ and $M_1(S'_i)=M_1(S_i)=B_{f_i}(\mathcal{G},s)$. These assumptions are trivially satisfied if $i=1$. We will show below that, at the end of $S'_i$, $G_{i}$ is entirely explored and if $i<|S|$ the size of $Sk$ is $0$. We will also show that if $i<|S|$ then $M_1(S'_{i+1})=B_{f_{i+1}}(\mathcal{G},s)$. We consider two cases.

The first case is when $\alpha f_i\geq 2$. We assume for simplicity that the number of edge traversals in $S_i$ is a positive multiple of $\lfloor\frac{ \alpha f_i}{2} \rfloor$.
As we will explain in detail, in this case the agent executes $S_i$ but interrupts it after each block of $\lfloor\frac{ \alpha f_i}{2} \rfloor$ edge traversals, except the last one, to make a ``return trip'' to node $s$ before resuming $S_i$ from where it was interrupted. The goal of these return trips is to satisfy the frequent return property. Once the agent has executed all instructions of $S_i$, it is either at the node containing the treasure or at node $s$. In the first case, we know that $i=|S|$ and $S'_i$ is simply over. In the second case $i<|S|$, but we do not have the guarantee that the size of $Sk$ is $0$. Hence, if the agent occupies node $s$ once it has executed all instructions of $S_i$, it then finishes $S'_i$ with what we call a \emph{close period} in which it executes in the reverse order some of the last edge traversals so that the size of $Sk$ becomes $0$ at the end of $S'_i$.

Denote by $v_k$ the node in which the $k$th interruption occurs, and by $P_k$ the path of length at most $\lfloor (1+\frac{\alpha}{2})f_i\rfloor$ from node $s$ to $v_k$ that is known by the agent when the interruption occurs. Note that $P_k$ necessarily exists in view of Lemma~\ref{lem:search1}, of the initial assumptions concerning $S'_i$ and of the fact that no edge traversal of $S_i$ has been skipped before the $k$th interruption. Also note that if there are several paths that can play the role of $P_k$, we simply choose the lexicographically smallest shortest path among them.

Each interruption is composed of two parts. In the first interruption, the first part consists in backtracking to node $s$ by executing in the reverse order the last $\lfloor \frac{\alpha f_i}{2} \rfloor$ edge traversals. The second part consists in going back to node $v_1$ using path $P_1$ to resume $S_i$. For the $k$th interruption with $k>1$, the first part consists in backtracking to node $s$ by executing in the reverse order the last $|P_{k-1}|+\lfloor \frac{\alpha f_i}{2} \rfloor$ edge traversals, and the second part consists in going back to node $v_k$ using path $P_k$ to resume $S_i$. Finally, the close period simply consists in backtracking to node $s$ by executing in the reverse order the last $\lfloor \frac{\alpha f_i}{2} \rfloor$ edge traversals if $S_i$ is made of only one block of $\lfloor\frac{ \alpha f_i}{2} \rfloor$ edge traversals. Otherwise, it consists in backtracking to node $s$ by executing in the reverse order the last $|P_{k^*-1}|+\lfloor \frac{\alpha f_i}{2} \rfloor$ edge traversals where $k^*=\frac{c_i}{\lfloor \frac{\alpha f_i}{2} \rfloor}$ is the number of blocks of $\lfloor\frac{ \alpha f_i}{2} \rfloor$ edge traversals in $S_i$.

It follows by induction on the number of interruptions that the size of $Sk$ is $0$ at the end of the first part of each interruption. Using this, the fact that $Sk$ is empty at the beginning of $S'_i$ and the fact that for every $1< k\leq \frac{c_i}{\lfloor \frac{\alpha f_i}{2} \rfloor}$, $|P_{k-1}|+\lfloor \frac{\alpha f_i}{2} \rfloor\leq \lfloor(1+\alpha)f_i\rfloor$, it follows that the frequent return property is satisfied during $S'_i$.

Moreover, it follows from the above explanation that at the end of $S'_i$, $G_{i}$ is entirely explored and the agent is at the node containing the treasure, if $i=|S|$. If $i<|S|$, it also follows that the size of $Sk$ is $0$ at the beginning of the next emulation $S'_{i+1}$, and $M_1(S'_{i+1})=M_1(S_{i+1})=B_{f_{i+1}}(\mathcal{G},s)$ because $M_2(S'_i)=M_2(S_i)=M_1(S_{i+1})$. Finally, concerning the cost of $S'_i$ observe that the number of interruptions is  $\frac{c_i}{\lfloor\frac{\alpha f_i}{2}\rfloor}-1$ and during each interruption as well as during the close period the agent makes at most $2\lfloor (1+\alpha)f_i\rfloor$ edge traversals. The cost of $S'_i$ is then upper bounded by $c_i+\frac{c_i}{\lfloor\frac{\alpha f_i}{2}\rfloor}2\lfloor (1+\alpha)f_i\rfloor\leq (1+\frac{2(1+\alpha)f_i}{\lfloor\frac{\alpha f_i}{2}\rfloor})c_i$. If $2\leq\alpha f_i<4$, then $\frac{2}{\alpha}\leq f_i<\frac{4}{\alpha}$, which implies that the cost is at most $(1+\frac{8(1+\alpha)}{\alpha})c_i$. Otherwise, $\alpha f_i\geq4$ and the cost is then upper bounded by $(1+\frac{2(1+\alpha)f_i}{\frac{\alpha f_i}{2}-1})c_i\leq (1+\frac{2(1+\alpha)}{\frac{\alpha}{2}-\frac{1}{f_i}})c_i$ which is also at most $(1+\frac{8(1+\alpha)}{\alpha})c_i$, as $\frac{1}{f_i}\leq\frac{\alpha}{4}$. Hence, the cost of $S'_i$ is $\mathcal{O}(c_i)$ as $\alpha$ is a constant, which concludes the first case.


The second case is when $\alpha f_i< 2$. Here, we could not apply the same strategy as that of the first case because we have $\lfloor\frac{ \alpha f_i}{2} \rfloor=0$. Consequently, we adopt a slightly different strategy in which the agent executes $S_i$ but interrupts it before each of its edge traversals. As explained in detail below, the $k$th interruption either consists of a return trip to node $s$ before resuming $S_i$ and making the $k$th edge traversal of $S_i$, or it consists in going to the node the agent should occupy at the end of the $k$th edge traversal of $S_i$ but without taking the corresponding edge: the agent then resumes $S_i$ as if it had just performed the $k$th edge traversal of $S_i$ (essentially it just makes some computations before interrupting again $S_i$ for the next edge traversal, if any). We will show that the latter situation will occur only when the ``skipped edge'' has already been traversed before by the agent. Once $S_i$ has been entirely processed, $S'_i$ is simply over if the agent is located at the node containing the treasure. Otherwise, the agent is at node $s$ and $i<|S|$. In this case, it executes (similarly as in the previous case) a close period in order to guarantee that the size of $Sk$ is $0$ at the end of $S'_i$.

Let us first focus on the interruptions. We denote by $(u_1,u_2,u_3,\ldots, u_{c_i+1})$ the sequence (with repetitions), in the chronological order, of the nodes that are visited during $S_i$, and by $(e_1,e_2,e_3,\ldots, e_{c_i})$ the sequence (with repetitions), in the chronological order, of the edges that are traversed during $S_i$. Consider the $k$th interruption occuring at node $u_k$ just before the $k$th edge traversal of $S_i$ and assume that at the beginning of this interruption, the property $H(k)$, consisting of the following three conditions, is satisfied: 
\begin{itemize}
\item The agent has made $D_k\leq f_i+1$ edge traversals since the last time when $Sk$ was empty (this could be the current time).
\item The sequence of edges $(e_1,e_2,\ldots, e_{k-1})$ has been previously explored by the agent.
\item The size of $Sk$ has been $0$ at least once during any previous block of $2(f_i+1)$ consecutive edge traversals and has never been greater than $f_i+1$.
\end{itemize}



Note that at the beginning of the first interruption, property $H(1)$ immediately holds. We will show below that property $H(k+1)$ is satisfied at the beginning of the $(k+1)$th interruption, if any.

In the $k$th interruption, the agent first checks whether it knows a path of length at most $f_i$ from node $s$ to node $u_k$. If this is the case, the agent executes in the reverse order the last $D_k$ edge traversals, at the end of which it is at node $s$ and $Sk$ is empty. Then, the agent comes back to $u_k$ using the known path of length at most $f_i$ from node $s$ to node $u_k$ (as when $\alpha f_i\geq 2$, if there are several such paths, the agent chooses the lexicographically smallest shortest among them). Once this is done, the interruption is over: the agent resumes $S_i$ and makes the $k$th edge traversal to reach node $u_{k+1}$. We can easily show that at the end of this edge traversal, and thus at the beginning of the next interruption if any, property $H(k+1)$ is satisfied.

So, assume that at the beginning of the $k$th interruption, the agent does not know a path of length at most $f_i$ from node $s$ to node $u_k$. In view of the fact that $D_k\leq f_i+1$, the shortest path from node $s$ to node $u_k$ that is known by the agent has actually length exactly $f_i+1$. Before explaining what the agent does, let us give some properties that necessarily hold in this situation. We have the following claim.
\vspace{-0.1cm}
\begin{claim}
$e_k$ belongs to $G_{i-1}$ or to the sequence $(e_1,e_2,\ldots, e_{k-1})$.
\end{claim}
\vspace{-0.1cm}
\begin{proofclaim}
Suppose by contradiction that the claim does not hold. Let $\mathcal{G}^*$ be the graph $\mathcal{G}$ with a midpoint $z$ added on edge $e_k$ and denote by $S^*_i$ the $i$th call to {\tt Search}$(\frac{\alpha}{2})$ made during an execution of {\tt TreasureHunt}$(\frac{\alpha}{2})$ from node $s$ in $\mathcal{G}^*$. Note that, in view of the execution of {\tt TreasureHunt}$(\frac{\alpha}{2})$ in $\mathcal{G}$, the definition of $\mathcal{G}^*$ and the assumption that the claim does not hold, $S^*_i$ indeed exists and at the beginning of $S^*_i$ variable $\mathcal{M}$ is $B_{f_i}(\mathcal{G}^*,s)=B_{f_i}(\mathcal{G},s)$. We know that just before making its $k$th edge traversal during $S_i$, the agent is at node $u_k$ and the shortest path that is known by the agent from node $s$ to node $u_k$ in $\mathcal{G}$ has length exactly $f_i+1$. Hence, the definition of $\mathcal{G}^*$ and the fact that $e_k$ does not belong to $G_{i-1}$ or to the sequence $(e_1,e_2,\ldots, e_{k-1})$, implies that during $S^*_i$, there is a time $t$ when the agent is at $z$ and the shortest path that is known by the agent from node $s$ to $z$ in $\mathcal{G}^*$ has length exactly $f_i+2$. However, by Lemma~\ref{lem:search1}, the length of the shortest path that is known by the agent from node $s$ to $z$ in $\mathcal{G}^*$ at time $t$ must have length at most $\max\{f_i+1,\lfloor (1+\frac{\alpha}{2})f_i \rfloor\}$, which is at most $f_i+1$ because $\frac{\alpha f_i}{2}<1$. This is a contradiction, which proves the claim.
\end{proofclaim}

From the above claim, it follows that at the beginning of the $k$th interruption the agent has already traversed edge $e_k$ before, and already knows which edge of $G_{i-1}$ or of $(e_1,e_2,\ldots, e_{k-1})$ corresponds to it. Thus, at the beginning of the $k$th interruption, the agent can already determine a path of length at most $f_i+1$ from node $s$ to node $u_{k+1}$ because in view of Lemma~\ref{lem:search1} it must know such a path when reaching node $u_{k+1}$ and because the traversal of $e_k$ does not bring extra topological information on $\mathcal{G}$.

Now we are are able to formulate what the agent does, when it has noticed that it does not know a path of length at most $f_i$ from node $s$ to node $u_k$. It executes in the reverse order the last $D_k$ edge traversals, at the end of which it is at node $s$ and $Sk$ is empty. Then, instead of coming back to $u_{k}$, it goes directly to node $u_{k+1}$ using the known path (highlighted in the previous paragraph) of length at most $f_i+1$ from node $s$ to node $u_{k+1}$. Once this is done, the interruption is over: the agent resumes $S_i$ and acts as if it had just traversed edge $e_k$ (as previously mentioned, it just performs some computations before interrupting again $S_i$ for the next edge traversal, if any). It follows that at the end of the interruption, and thus at the beginning of the following one if any, $H(k+1)$ is satisfied. 
We have shown by induction on $k$ that, at the beginning of the $k$th interruption, for any $k\geq 1$, the property $H(k)$ is satisfied.  This closes the description of the interruptions.

It remains to deal with the close period. At the beginning of it, property $H(c_i+1)$ is satisfied, which implies that the agent has performed $D_{c_i+1}\leq f_i+1$ edge traversals since the last time when $Sk$ was empty. Hence, during the close period, the agent simply executes in the reverse order the last $D_{c_i+1}$ edge traversals, at the end of which $Sk$ is empty. In view of this, of the fact that $Sk$ is empty at the beginning of $S'_i$, and of property $H(c_i+1)$, the frequent return property is satisfied during $S'_i$.

It follows from the above explanation that at the end of $S'_i$, $G_{i}$ is entirely explored and the agent is at the node containing the treasure if $i=|S|$. If $i<|S|$, it also follows that the size of $Sk$ is $0$ at the beginning of the next emulation $S'_{i+1}$, and $M_1(S'_{i+1})=M_1(S_{i+1})=B_{f_{i+1}}(\mathcal{G},s)$ because $M_2(S'_i)=M_2(S_i)=M_1(S_{i+1})$. Finally, concerning the cost of $S'_i$, observe that the number of interruptions is  $c_i$ and during the close period as well as during each interruption the agent makes at most $2(f_i+1)$ edge traversals. The cost of $S'_i$ is then upper bounded by $2(f_i+1)c_i+2f_i+2$ which is at most $2(\frac{2}{\alpha}+1)c_i+\frac{4}{\alpha}+2$, as $f_i<\frac{2}{\alpha}$ in the currently analysed case. Hence, the cost of $S'_i$ is $\mathcal{O}(c_i)$. This concludes the second case and thus concludes the proof of the theorem.
\end{proof}

\vspace{-0.5cm}
\section{Conclusion}

We presented treasure hunt algorithms working at cost ${\cal O}(e(d)\log d)$ for the unrestricted, fuel-restricted and rope-restricted models. Hence our algorithms are nearly linear in $e(d)$ (thus refuting the conjecture from \cite{ABRS}) and at the same time almost optimal, as cost $\Theta(e(d))$ cannot be beaten in general. The natural open problem is whether it is possible to get rid of the factor ${\cal O}(\log d)$ that separates us from optimal complexity of treasure hunt.


\vspace{-0.2cm}
\begin{thebibliography}{12}

\bibitem{AG}
S.Alpern, S.Gal, The Theory of Search Games and Rendezvous, Kluwer Academic Publications, 2003.

\bibitem{AAD}
S. Angelopoulos, D. Arsenio, C. Durr, Infinite linear programming and online searching with turn cost, Theoretical Computer Science  670 (2017), 11-22.

\bibitem{ABRS}
B. Awerbuch, M. Betke, R.L. Rivest, M. Singh, Piecemeal graph exploration by a mobile robot, Information and Computation 152 (1999), 155-172.


\bibitem{BCR}
R. Baeza-Yates, J. Culberson, J. Rawlins, Searching the plane,
Information and Computation 106 (1993), 234-252.

\bibitem{Bec1964}
			A. Beck,
			On the linear search problem,
			Israel Journal of Mathematics, 2 (1964), 221-228.

			\bibitem{BN1970}
			A. Beck, D.J. Newman,
			Yet more on the linear search problem,
			Israel Journal of Mathematics, 8 (1970), 419-429.

			\bibitem{Bel1963}
			R. Bellman,
			An optimal search problem,
			SIAM Review 5 (1963), 274.

\bibitem{DCD}
P. Dasgupta, P.P. Chakrabarti, S.C. DeSarkar, Agent searching in a tree and the optimality of
iterative deepening, Artificial Intelligence 71 (1994), 195-208.

\bibitem{DCS1995}
			P. Dasgupta, P.P. Chakrabarti, S.C. DeSarkar,
			A Correction to ``Agent Searching in a Tree and the Optimality of Iterative Deepening",
			Artificial Intelligence 77 (1995), 173-176.

\bibitem{DFG2006}
			E.D. Demaine, S.P. Fekete, S. Gal,
			Online searching with turn cost,
			Theoretical Computer Science 361 (2006), 342-355.



\bibitem{DKK}
			C.A. Duncan, S.G. Kobourov, V.S.A. Kumar,
			Optimal constrained graph exploration,
			ACM Trans. Algorithms 2 (2006), 380-402.

			\bibitem{FKK+2008}
			R. Fleischer, T. Kamphans, R. Klein, E. Langetepe, G. Trippen,
			Competitive Online Approximation of the Optimal Search Ratio,
			SIAM J. Comput. 38 (2008), 881-898.





\bibitem{FHGTM}
G. M. Fricke, J. P. Hecker, A. D. Griego, L. T. Tran, M.E. Moses, A Distributed Deterministic Spiral Search Algorithm for Swarms, Proc. IEEE/RSJ International Conference on Intelligent Robots and Systems (IROS 2016), 4430-4436.

\bibitem{Gal2013}
			S. Gal,
			Search Games: A Review,
			Search Theory: A Game Theoretic Perspective (2013), 3-15.


\bibitem{GK2010}
			S.K. Ghosh, R. Klein,
			Online algorithms for searching and exploration in the plane,
			Computer Science Review 4 (2010), 189-201.







\bibitem{JL}
A. Jez, J. Lopuszanski, On the two-dimensional cow search problem, Information Processing Letters 109 (2009), 543 - 547.

\bibitem{KRT}
M.Y.Kao, J.H. Reiff, S.R. Tate, Searching in an unknown environment: an optimal randomized algorithm for the cow-path problem,
Information and Computation 131 (1997), 63-80.

\bibitem{KZ2011}
			D.G. Kirkpatrick, S. Zilles,
			Competitive search in symmetric trees,
			Proc. 12th International Symposium on Algorithms and Data Structures (WADS 2011), 560-570.


\bibitem{KKKS}
D. Komm, R. Kralovic, R. Kralovic, J. Smula,
Treasure hunt with advice, Proc. 22nd International Colloquium on Structural Information and Communication Complexity  (SIROCCO 2015), 328-341.


\bibitem{La}
E. Langetepe, On the Optimality of Spiral Search, Proc. 21st Ann. ACM-SIAM Symp. Disc. Algor. (SODA 2010), 1-12.

\bibitem{La2}
E. Langetepe, Searching for an axis-parallel shoreline, Theoretical Computer Science 447 (2012), 85-99.

\bibitem{LS2001}
			A. Lopez-Ortiz, S. Schuierer,
			The ultimate strategy to search on m rays?,
			Theoretical Computer Science 261 (2001), 267-295.


\bibitem{MP}
A. Miller, A. Pelc, Tradeoffs between cost and information for rendezvous and treasure hunt, Journal of Parallel and Distributed Computing 83 (2015), 159-167.



\bibitem{Pe}
A. Pelc, Reaching a target in the plane with no information, Information Processing Letters 140 (2018), 13-17.


\bibitem{Sch2001}
			S. Schuierer,
			Lower bounds in on-line geometric searching,
			Computational Geometry 18 (2001), 37-53.


\bibitem{TSZ}
A. Ta-Shma, U. Zwick,
Deterministic rendezvous, treasure hunts and strongly universal exploration sequences.
ACM Transactions on Algorithms 10 (2014), 12:1-12:15.





\end{thebibliography}
\end{document}